\documentclass[conference]{IEEEtran}
\usepackage{fancyhdr} 
\usepackage[utf8]{inputenc}
\usepackage{cite}

\usepackage{floatrow}

\usepackage{comment}
\usepackage[font=small]{caption}
\usepackage{times}
\usepackage{epsfig}
\usepackage{graphicx, subfig}
\usepackage{amsmath}
\usepackage{amssymb}
\usepackage{amsthm}
\usepackage{url}
\usepackage{float}
\usepackage{xcolor}
\usepackage{algorithm,algpseudocode}
\usepackage{listings}
\usepackage{pifont}
\usepackage{multirow}
\usepackage{tikz}
\newtheorem{theorem}{Theorem}[]
\newcommand*\circled[1]{\tikz[baseline=(char.base)]{
            \node[shape=circle,fill,inner sep=1pt] (char) {\footnotesize \textcolor{white}{#1}};}}

\newcommand{\gsag}{\textsc{c-saw}}
\newcommand{\ipdps}{Zeng et al.}
\definecolor{err}{RGB}{249,159,159}

\def\HiLi{\leavevmode\rlap{\hbox to \hsize{\color{err}\leaders\hrule height .8\baselineskip depth .5ex\hfill}}}

\newcommand{\new}[1]{{\color{black}{#1}}}
\newcommand{\cready}[1]{{\color{black}{#1}}}

\newcommand{\hang}[1]{\textcolor{red}{[\small Hang: ~#1~]}}

\begin{document}

\date{}

\title{\gsag: A Framework for Graph Sampling and Random Walk on GPUs}

\author{Santosh Pandey\hspace{.5in}Lingda Li$^+$\hspace{.5in}Adolfy Hoisie$^+$\hspace{.5in}Xiaoye S. Li$^*$\hspace{.5in}Hang Liu\\
Stevens Institute of Technology\hspace{.1in}$^+$Brookhaven National Laboratory\hspace{.1in}$^*$Lawrence Berkeley National Laboratory\vspace{-.5in}}

\maketitle
\thispagestyle{fancy}
\lhead{}
\rhead{}
\chead{}
\lfoot{\footnotesize{
SC20, November 9-19, 2020, Is Everywhere We Are
\newline 978-1-7281-9998-6/20/\$31.00 \copyright 2020 IEEE}}
\rfoot{}
\cfoot{}
\renewcommand{\headrulewidth}{0pt}
\renewcommand{\footrulewidth}{0pt}

\providecommand{\keywords}[1]
{
  \small	
  \textbf{\MakeUppercase{Keywords---}} #1
}

\begin{abstract}


Many applications require to learn, mine, analyze and visualize large-scale graphs.
These graphs are often too large to be addressed efficiently using conventional graph processing technologies.
Fortunately, recent research efforts find out {\em graph sampling} and {\em random walk}, which significantly reduce the size of original graphs, can benefit the tasks of learning, mining, {analyzing} and visualizing large graphs by capturing the desirable graph properties.
This paper introduces {\gsag}, the first framework that a\underline{c}celerates \underline{S}ampling and R\underline{a}ndom \underline{W}alk framework on GPUs. Particularly, {\gsag} makes three contributions:
First, our framework provides a generic API which allows users to implement a wide range of sampling and random walk algorithms with ease.
Second, offloading this framework on GPU, we introduce warp-centric parallel selection, and two novel optimizations for collision migration.
Third, towards supporting graphs that exceed the GPU memory capacity, we introduce efficient data transfer optimizations for out-of-memory and multi-GPU sampling, such as {workload-aware} scheduling and batched {multi-instance} sampling.
Taken together, our framework constantly outperforms the state of the art projects in addition to the capability of supporting a wide range of sampling and random walk algorithms. 
\end{abstract} \hspace{10pt}


\section{Introduction}

Graph is a natural format to represent relationships that are prevalent in a wide range of real-world applications, such as, material/drug discovery \cite{you2018graph}, web-structure~\cite{kumar2002web}, social network~\cite{garton1997studying}, protein-protein interaction~\cite{von2002comparative}, knowledge graphs~\cite{popping2003knowledge}, among many others.
Learning, mining, analyzing and visualizing graphs is hence of paramount value to \cready{our society}. 
However, as the size of the graph continues to grow, the complexity of handling those graphs also soars. In fact, large-scale graph analytics is deemed as a grand challenge that draws a great deal of attention. Popular evidences are Graph 500~\cite{murphy2010introducing} and GraphChallenge~\cite{graphchallenge}.

Fortunately, recent research efforts find out {\em graph sampling} and {\em random walk}, which significantly reduce the size of \cready{the} original graphs, can benefit learning, mining and \cready{analyzing} large graphs, by capturing the desirable graph properties~\cite{huang2018adaptive,gao2018large,chen2017stochastic}. For instance, {\ipdps~\cite{zeng2019accurate}}, GraphSAINT~\cite{zeng2019graphsaint}, GraphZoom~\cite{deng2019graphzoom}, Pytorch-biggraph~\cite{lerer2019pytorch} and \cready{Deep Graph Library (DGL)}~\cite{wang2019deep} manage to learn from the sampled graphs and arrive at vertex embeddings that are either similar or better than directly learning on the original gigantic graphs~\cite{deng2019graphzoom}. Weisfeiler-Lehman Algorithm~\cite{shervashidze2011weisfeiler} exploits graph sampling to find isomorphic graphs. {Furthermore, various random walk methods are used to generate vertex ranking and embedding in a graph~\cite{perozzi2014deepwalk,page1999pagerank,grover2016node2vec,kyrola2013drunkardmob}.} Sampling and random walk can also help classical graph computing algorithms, such as BFS~\cite{korf2005large} and PageRank~\cite{page1999pagerank}.

Despite great importance,
limited efforts have been made to deploy graph sampling and random walk algorithms on GPUs which come with tempting computing, data access capabilities and ever-thriving community~\cite{gao2018large}.
This paper finds three major challenges that prevent this effort.

\vspace{.1in}
First, although there is a variety of platforms to accelerate traditional graph processing algorithms on GPUs~\cite{liu2015enterprise,wang2016gunrock,liu2019simd,gaihre2019xbfs}, graph sampling and random walk pose unique challenges. Unlike traditional graph algorithms which often treat various vertices and edges similarly and focus on optimizing the operations on the vertex or edge, sampling and random walk algorithms center around how to select a subset of vertex or edge \cready{based upon a bias} (Section \ref{sect:background:ns}). Once selected, the vertex is merely visited again.
Consequently, {\em how to efficiently select the vertices of interest which is rarely studied by traditional algorithms becomes the core of sampling and random walk.}
This process needs to construct and potentially update the selection \cready{bias} repeatedly which is very expensive hence significantly hampers the performance.  

\vspace{.1in}
Second, it is difficult to arrive at a GPU-based framework for various graph sampling and random walk algorithms that address the needs of vastly different applications. Particularly, there exists a rich body of graph sampling and random walk algorithms (detailed in Section~\ref{sect:background:samplerw}), deriving the common functionalities for a framework and exposing different needs as user programming interface is a daunting task. And \cready{offloading} this framework on GPU to enjoy the unprecedented computing and bandwidth capability yet hiding the GPU programming complexity further worsens the challenge. 


\begin{table*}[t]
\new{{\fontsize{8}{10}\selectfont
\centering
\begin{tabular}{|c|c|c|c|c|c|}

\hline
\multicolumn{2}{|c|}{\multirow{3}{*}{\begin{tabular}[c]{@{}c@{}}Bias\\criterion\end{tabular}}} & \multicolumn{4}{c|}{\# of neighbors (NeighborSize)}                                   \\ \cline{3-6} 
\multicolumn{2}{|c|}{}                                                                              & \multicolumn{2}{c|}{Per layer} & \multicolumn{2}{c|}{Per vertex} \\ \cline{3-6} 
\multicolumn{2}{|c|}{}                                                                              & 1   & $>1$    &Constant       & Variable         \\ \hline
\multicolumn{2}{|c|}{Unbiased}    & \begin{tabular}[c]{@{}l@{}}Simple random walk, metropolis hasting random walk,\\ random walk with Jump, random walk with restart\end{tabular}  &   & Unbiased neighbor sampling & \begin{tabular}[c]{@{}l@{}}\new{Forest fire sampling},\\ Snowball sampling\end{tabular}                 \\ \hline
\multirow{2}{*}{Biased}                                  
& Static    &Biased random walk     & Layer sampling    &Biased neighbor sampling   &        \\ \cline{2-6} 
& Dynamic   & Multi-dimensional random walk, Node2vec     &   &  &          \\ \hline
\end{tabular}
}
}

\vspace{-.1in}
\caption{\cready{The design space of traversal based sampling and random walk algorithms.}
\vspace{-.2in}
}
\label{tbl:samplespace}

\end{table*}

\vspace{.1in}Third, an extremely large graph, which drives the needs of graph sampling and random walk, usually goes beyond the size of GPU memory. While there exists an array of solutions for GPU-based large graph processing, namely, unified memory~\cite{gera2020traversing}, topology-aware partition~\cite{karypis1998fast} and vertex-range based partitions~\cite{guattery1995performance}, graph sampling and random walk algorithms, which require all the neighbors of a vertex to present in order to compute the selection probability, exhibit stringent requirement on the partitioning methods.  
In the meantime, the asynchronous and out-of-order nature of graph sampling and random walk provides some unique optimization opportunities for {\em out-of-memory sampling}, which are neither shared nor explored by traditional out-of-memory systems.

This work advocates {\gsag}, to the best of our knowledge, the first GPU-based framework that addresses all the three aforementioned challenges and supports a wide range of sampling and random walk algorithms. Taken together, {\gsag} significantly outperforms the state of the art systems that support either part of sampling or random walk algorithms.  
The contributions of this paper are as follows:

\vspace{.1in}
\begin{itemize}
\item We propose a generic framework which allows end users to express a large family of sampling and random walk algorithms with ease (Section \ref{sect:arch}).

\vspace{.05in}
\item We implement efficient GPU sampling with novel techniques.
Our techniques parallelize the vertex selection on GPUs, with efficient algorithm and system optimizations for vertex collision migration (Section~\ref{sect:single}).

\vspace{.05in}
\item We propose {asynchronous} designs for sampling and random walk, which optimizes the data transfer efficiency for graphs that exceed the GPU memory capacity.
\new{We further scale {\gsag} to multiple GPUs} (Section \ref{sect:outmem}).
\end{itemize}
\vspace{.1in}

The remainder of this paper goes as follows: Section~\ref{sect:background} presents the background. Section~\ref{sect:arch} outlines the Application Programming Interface (API) and Sections~\ref{sect:single} and~\ref{sect:outmem} optimize {\gsag}. Section~\ref{sect:experiment} presents the evaluation results.
Section~\ref{sect:related} discusses the related works and Section \ref{sect:conclusion} concludes.

\section{Background}
\label{sect:background}

\begin{figure*}[hbt!]
	\centering
	\includegraphics[width=\linewidth]{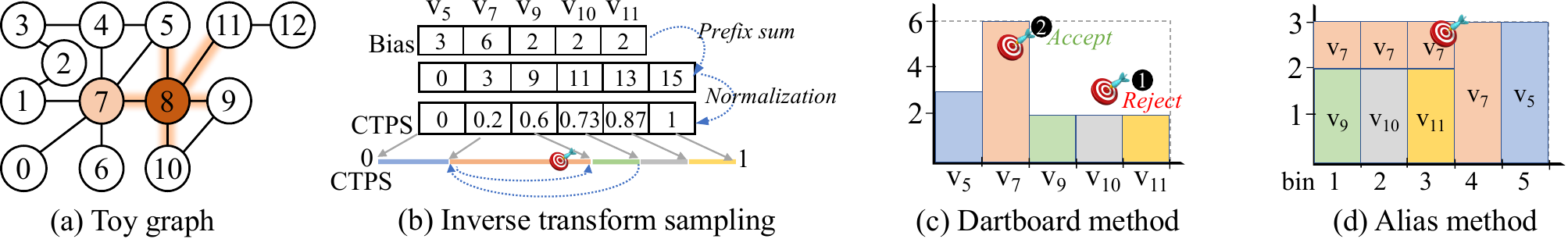}
	\caption{Example of graph sampling and vertex selection techniques. (a) A toy graph example to select a neighbor of $v_8$ ($v_5, v_7, v_9, v_{10}, v_{11}$), assuming the bias of a neighbor is defined as its degree. (b) Inverse Transform Sampling which does a binary search on a 1-D space to select $v_7$. (c) Dartboard method that rejects \protect\circled{1} and accepts \protect\circled{2} ($v_7$). (d) Alias method that selects $v_7$.
	}
	\label{fig:example}
\end{figure*}

\subsection{Graph Sampling \& Random Walk Variations}
\label{sect:background:samplerw}
\new{
This section presents the required background for various graph sampling and random walk algorithms~\cite{leskovec2006sampling}.}
Graph sampling refers to the random exploration of a graph, which results in a subgraph 
of the original graph.

\new{\textbf{One Pass Sampling}
only goes through the original graph once to extract a sample.
Random node and random edge sampling belong to this category \cite{leskovec2006sampling}.
They select a subset of vertices/edges in original graph uniformly and randomly.
}

\textbf{Traversal based Sampling}
often traverses the graph in a Breath-First Search manner to better preserve the properties of original graphs \cite{hu2013survey}.
Traversal based sampling follows \textit{sampling without replacement} methodology, i.e., it avoids sampling the same vertex more than once.

As shown in Table \ref{tbl:samplespace},
traversal based sampling algorithms are categorized based upon the number of sampled neighbors, called {\textit{NeighborSize}}, and the criterion to select neighbors, which is referred to as {\em bias}.
Snowball sampling~\cite{stivala2016snowball} initiates the sample using a set of uniformly selected seed vertices.
Iteratively, it adds all neighbors of every sampled vertex into the sample, until a required depth is reached.
Neighbor sampling~\cite{neighborsampling} 
samples a constant number of neighbors per vertex. The sampling could be \cready{either} biased or unbiased.
Forest fire sampling~\cite{leskovec2006sampling} can be regarded as a probabilistic version of neighbor sampling, which selects a variable number of neighbors for each vertex based on a burning probability.
Unlike neighbor and forest fire sampling, which select neighbors for each vertex independently, layer sampling \cite{gao2018large} samples a constant number of neighbors for all vertices present in the frontier in each round.
It repeats this process until a certain depth is reached.

\new{
\textbf{Random Walk} simulates a stochastic process of traversing the graph to form a path of connected vertices. The length of path is constrained by a user given sampling budget.
Random walk can be viewed as a special case of sampling when only one neighbor is sampled at a step with the salient difference lies in that random walk allows repeated appearance of a vertex while sampling does not.
Table \ref{tbl:samplespace} summarizes the design space of random walk algorithms.
}

Similar to traversal based sampling, random walk algorithms use {\em bias} to decide the probability of selecting a certain neighbor.
For unbiased simple random walk, the bias is uniform for all neighbors, i.e., every neighbor has the same chance to be selected.
Deepwalk~\cite{perozzi2014deepwalk} and metropolis hasting random walk \cite{li2015random} are two examples of unbiased random walk. While Deepwalk samples neighbors uniformly, meropolis hasting random walk decides to either explore the sampled neighbor or choose to stay at the same vertex based upon the degree of source and neighbor vertices.

For a biased random walk, the bias varies across neighbors.
Furthermore, depending on how to decide the bias, biased random walks are classified \cready{into} static random walks \cready{and} dynamic random walks.
For static random walk, the bias is determined by the original graph structure and does not change at runtime. Biased Deepwalk~\cite{cochez2017biased} is an example of static random walk which extends the original Deepwalk algorithm. The degree of each neighbor is used as its bias.

\new{Since a simple random walk may get stuck 
locally, 
random walk with jump~\cite{tzevelekas2010random}, random walk with restart~\cite{tong2006fast} and multi-independent random walk~\cite{hu2013survey} are introduced. 
Particularly, random walk with jump 
jumps to a random vertex under a certain probability.
Random walk with restart
jumps to a predetermined vertex.
Multi-independent random walk performs multiple instances of random walk independently.
}

For dynamic random walks, the bias depends upon the runtime states. 
Node2vec~\cite{grover2016node2vec} and multi-dimensional random walk (a.k.a. frontier sampling)~\cite{ribeiro2010estimating} belong to this genre.
Node2vec is an advanced version of Deepwalk which provides more control to the random walk.
The bias of a neighbor depends upon the edge weight and its distance from the vertex explored at preceding step.
In multi-dimensional random walk, a pool of seed vertices are selected at the beginning.
At each step, multi-dimensional random walk explores one vertex $v$ from the pool based on their degrees.
One random neighbor of $v$ is added to the pool to replace $v$.
This process repeats until a desired number of vertices are sampled.

\new{
\textbf{Summary.}
Traversal based sampling and random walk are widely used and share two core similarities: 1) they are based on graph traversal, and 2) they selectively sample vertices based on biases (detailed in Section \ref{sect:background:ns}).
Their difference is the number of sampled neighbors, as shown in Table \ref{tbl:samplespace}.
In the rest of this paper, we use {\bf graph sampling to refer to both traversal based sampling and random walk}, unless explicitly specified.
}

\subsection{Bias based Vertex Selection}
\label{sect:background:ns}

\new{This section discusses the key challenge of graph sampling:} to select vertices based on user defined biases, i.e., {\em bias based vertex selection}.
As discussed in Section \ref{sect:background:samplerw}, all sampling algorithms involve the process of picking up a subset of vertices from a candidate pool of vertices.
For \cready{unbiased graph sampling}, the selection is straightforward: 
one can generate a random integer in the range of 1 to the candidate count and use it to select a vertex.
Vertex selection is more challenging \cready{biased graph sampling}.
Given certain biases, we need to calculate the probability of selecting a certain vertex, which is called {\em transition probability}.
Theorem \ref{thm:tp} gives the formula to calculate transition probabilities from biases.

\begin{theorem}
\label{thm:tp}
Let \new{vertices $v_1, v_2, ..., v_{n}$} be the $n$ candidates, and the transition probability of $v_k$, i.e., $t_k$, be proportional to the {\em bias} $b_{k}$.
Then, one can formally obtain 
\new{$t_k=\frac{b_{k}}{\sum_{i=1}^{n} b_{i}}$}.
\end{theorem}

Theorem \ref{thm:tp} underscores that \textit{bias} is the key to calculate transition probability.
All popular vertex selection algorithms -- inverse transform sampling~\cite{olver2013fast}, dartboard \cite{yang2019knightking}, and alias method~\cite{walker1977efficient,li2014reducing} -- obey this rule.

The key idea of inverse transform sampling is to generate the cumulative distribution function of the transition probability.
Fig.~\ref{fig:example}(b) shows an example.
First, inverse transform sampling computes the prefix sum of biases of candidate vertices, to get an array $S$, where \cready{$S_m = \sum_{i=1}^{m} b_{i-1}$ ($1 \le m \le n+1)$ and $n$= total \# of candidate vertices}.
In Fig.~\ref{fig:example}(b), $S = \{0, 3, 9, 11, 13, 15\}$.
Then $S$ is normalized using \cready{$S_{n+1}$}, to get array $F$, where \cready{$F_m = S_m/S_{n+1}\ (1 \le m \le n+1)$}.
$F = \{0, 0.2, 0.6, 0.73, 0.87, 1\}$ in Fig.~\ref{fig:example}(b).
In this way, the transition probability of $v_k$ can be derived with $F$, because

\setlength{\belowdisplayskip}{0pt} \setlength{\belowdisplayshortskip}{0pt}
\setlength{\abovedisplayskip}{0pt} \setlength{\abovedisplayshortskip}{0pt}

\new{{\footnotesize
\begin{equation}
\label{equ:tp}
\begin{aligned}
t_k&=\frac{b_{k}}{\sum_{i=1}^{n} b_{i}}=\frac{\sum_{i=1}^{k} b_{i}-\sum_{i=1}^{k-1} b_{i}}{\sum_{i=1}^{n} b_{i}}\\
&=\frac{S_{k}-S_{k-1}}{S_n}
=\frac{S_{k}}{S_n}-\frac{S_{k-1}}{S_n}=F_{k}-F_{k-1}.
\end{aligned}
\end{equation}
}
}

\noindent We call the array of $F$ \textbf{Cumulative Transition Probability Space (CTPS)}.
To select a neighbor, inverse transform sampling generates a random number $r$ in the range of (0,1), and employs a binary search of $r$ over the CTPS.
Assuming $r=0.5$ in Fig. \ref{fig:example}(b), it falls between \cready{$F_2 = 0.2$ and $F_3 = 0.6$.}
As a result, the \cready{second} candidate $v_7$ is selected on the CTPS.
When implemented sequentially, ITS has the computational complexity of $O(n)$, determined by the prefix sum calculation.

Dartboard \cite{yang2019knightking} uses 2D random numbers to select/reject vertices.
As shown in Fig.~\ref{fig:example}(c), we build a 2D board using the \cready{bias of each vertex as a bar}, and then throw a dart to the 2D board \cready{formed by two random numbers}.
If it does not hit any bar (e.g., \circled{1}), we reject the selection and throw another dart, until a bar is hit (e.g., \circled{2}).
This method may require many trials before \cready{picking up a vertex} successfully, especially for scale-free graphs where a few candidates have much larger biases than others.
Similar to dartboard, the alias method~\cite{li2014reducing} also uses a 2D \cready{board}.
To avoid rejection, the alias method converts the sparse 2D board into a dense one as shown in Fig. \ref{fig:example}(d).
It breaks down and distributes large biases across bins on the x axis, 
with the guarantee that a single bin contains at most two vertices.
The drawback of alias method is its high preprocessing cost to break down and distribute biases, which is not suitable for GPUs.

\section{$\gsag$ Architecture}
\label{sect:arch}

\vspace{0.1in}
\subsection{Motivation}
\label{sect:background:moti}

\new{
\textbf{Need for Generic Sampling Framework.} 
After sifting across numerous graph analytical frameworks (detailed in Section~\ref{sect:related}), 
we find the need of a new framework for graph sampling, because
\textit{sampling algorithms pose distinct needs on both the framework design and APIs}. 
For framework design, several sampling algorithms, e.g., layer sampling, require the information beyond a vertex and its neighbors for computing, which postulates hardship for traditional vertex-centric frameworks that limit the view of a user to a vertex and its 1-hop neighbors. 
When it comes to API design, \textit{\textbf{bias} is the essence of sampling and random walk}.
In comparison, traditional graph frameworks focus upon the operators that alter the information on an edge or a vertex, e.g., minimum operator in single source shortest path. We also notice recent attempts, e.g., KnightKing~\cite{yang2019knightking} and GraphSAINT~\cite{zeng2019graphsaint}, but they cannot support both sampling and random walk algorithms.}

\textbf{Need for Sampling and Random Walk on GPUs.}
For sampling, short turnaround time is the key.
It is also the root cause of the invention of sampling~\cite{zeng2019accurate,lofgren2014fast}. 
The good news is that GPU is a proven vehicle to drive an array of graph algorithms beyond their performance ceiling~\cite{liu2015enterprise,liu2016ibfs,wang2016gunrock,liu2019simd,gaihre2019xbfs,pandey2019h,bisson2017high}, thanks to the unprecedented computing capability and memory bandwidth~\cite{keckler2011gpus}. When it comes to sampling which are much more random than traditional graph algorithms, GPUs will best CPU at even larger margins because extreme randomness puts the large caches of CPU in vein.

\subsection{{\gsag}: A Bias-Centric Sampling Framework}
\label{sect:arch:overview}

{\gsag} offloads sampling and random walk on GPUs with the goal of a \textit{simple} and \textit{expressive} API and a \textit{high performance} framework.
Particularly, \textit{simple} means the end users can program {\gsag} without knowing the GPU programming syntax.
\textit{Expressiveness} requires {\gsag} to not only support the known sampling algorithms discussed in Section~\ref{sect:background:samplerw}, but also prepare to support emerging ones.
\textit{High performance} targets the framework design. 
That is, the programming simplicity does not prevent {\gsag} from exploring major GPU and sampling related optimizations. 

\begin{figure}[t]
	\centering
	\includegraphics[width=0.89\textwidth]{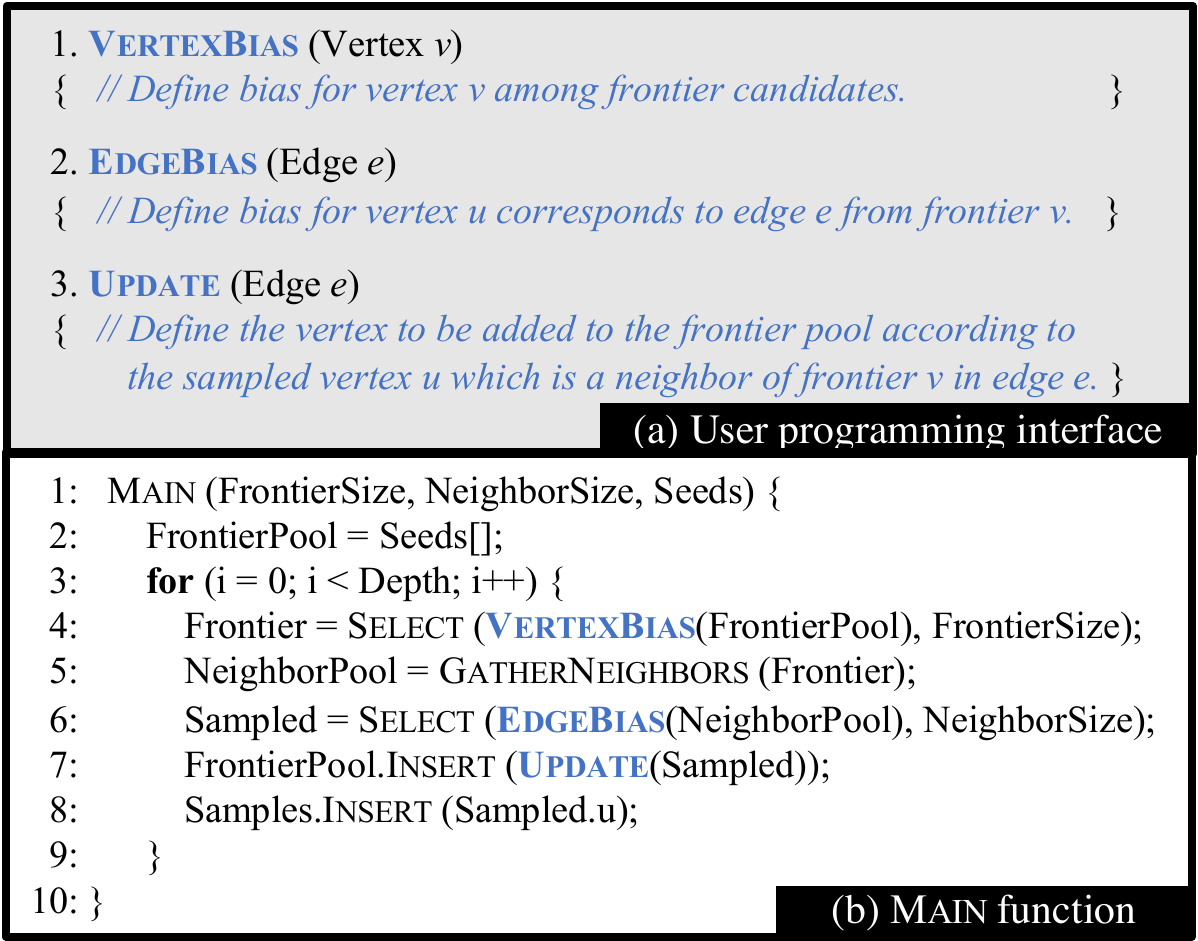}
	\caption{\new{\gsag~framework and API functions.}
	}
	\label{fig:api}
\end{figure}

{\gsag} encompasses two types of user involvements, i.e., parameter and API based options.  The parameter-based option only needs a choice from the end users thus is simple, e.g., deciding the number of selected frontier vertices ($FrontierSize$ in line 4 of Fig. \ref{fig:api}(b)) and neighbors ($NeighborSize$ in line 6).
API based involvement, in contrast, provides more expressiveness to users. Particularly, {\gsag} offers three user defined API functions \cready{as shown in Fig.~\ref{fig:api}(a)}, \cready{most} of which surround bias, that is, \textsc{VertexBias}, \textsc{EdgeBias}, and \textsc{Update}.
We will discuss the design details of these API functions in Section \ref{sect:arch:api}.

Fig. \ref{fig:api}(b) gives an overview of the {\gsag} algorithm.
Particularly, bias based vertex selection occurs in two places: to select frontier vertices from a pool (line 4), and to select the neighbors of frontier vertices (line 6).
While the latter case is required by all graph sampling algorithms, the former becomes essential when users want to introduce more randomness, such as \new{multi-dimensional random walk}.


\begin{figure}[t]
	\centering
	\includegraphics[width=0.89\textwidth]{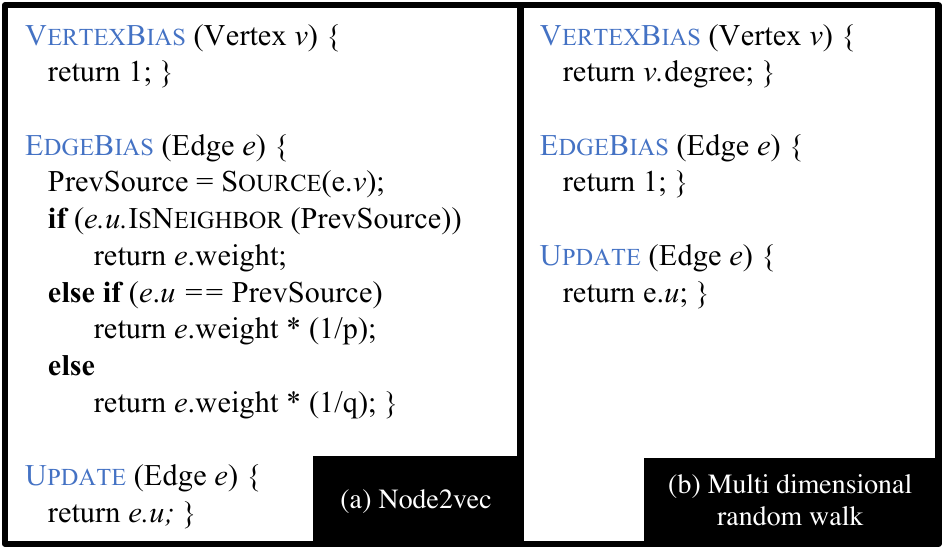}
	\vspace{-.1in}
	\caption{\new{Implementing two sampling algorithms with \gsag~API.}
	}
	\label{fig:api_example}
\end{figure}

In the beginning, the frontier \cready{$FrontierPool$} is initialized with a set of seed vertices (line 2).
Sampling starts from these seeds until reaching the desired depth (line 3).
In each iteration of the while loop, first, \textbf{\textsc{VertexBias}} is called on the \cready{$FrontierPool$} to retrieve the bias for each candidate vertex.
\textsc{Select} method uses the biases provided by \textbf{\textsc{VertexBias}} to choose $FrontierSize$ vertices as the current frontier (line 4).
Next, all neighbors of the frontier vertices are collected in the $NeighborPool$ using the \textsc{GatherNeighbors} method (line 5).
For these neighbors, we first define their biases using the \textbf{\textsc{EdgeBias}} method.
Similarly, \textsc{Select} method uses the biases to choose $NeighborSize$ neighbors from the $NeighborPool$ (line 6).
From the selected neighbors, \textbf{\textsc{Update}} is used to \cready{pick} new vertices for the \cready{$FrontierPool$} (line 7).
The selected neighbors are also added to the final sample \cready{list $Sampled$} (line 8) before we move forward to the next iteration.


\subsection{{\gsag} API}
\label{sect:arch:api}

\textbf{\textsc{VertexBias}} defines the bias associated with a candidate vertex of the \cready{$FrontierPool$}.
We often use the pertinent property of vertex to derive the bias. Equation~(\ref{equation:vertexbias}) formally defines the bias for each vertex $v$ in the \cready{$FrontierPool$}. 
\new{We apply function $f_{vBias}$ over the property of $v$ to define the associated bias.}
\new{\begin{equation}
    \label{equation:vertexbias}
     \textbf{\textsc{VertexBias}}\underset{v~{\in}~FrontierPool}{\longleftarrow} f_{vBias} (v).
\end{equation}}

Using multi-dimensional random walk as an example, it uses the vertex degree as a bias for the vertex of interest.


\textbf{\textsc{EdgeBias}} defines the bias of each neighbor in the \cready{\textit{NeighborPool}}.
It is named as {\textsc{EdgeBias}} because every neighboring vertex is associated with an edge.
While, again, any static or dynamic bias is applicable, a typical bias is induced from the properties of \cready{the associated edge}.
Equation~(\ref{equation:edgebias}) defines \textsc{EdgeBias} formally.
Let $v$ be the source vertex of $u$.
\new{Assuming edge $e=(v, u)$ carries the essential properties of $v$, $u$ and $e$, we arrive at the following edge bias:
\begin{equation}
{
\label{equation:edgebias}
     \textbf{\textsc{EdgeBias}} \underset{e~{\in}~NeighborPool}{\longleftarrow} f_{eBias} (e)
}
\end{equation}
}

\textbf{\textsc{Update}} decides the vertex that should be added to the \cready{$FrontierPool$} based on the sampled neighbors.
It can return any vertex to provide maximum flexibility.
For instance, this method can be used to filter out vertices that have been visited before for most traversal based sampling algorithms.
Whereas for random walk, this method can be used to implement the jump or restart action in the random walk with jump and with start, respectively. Equation~(\ref{equation:update}) quantifies this method, \new{where we will decide whether to add the sampled vertex $u$, a neighbor of frontier $v$ from edge $e$ into \cready{\textit{FrontierPool}} based upon the properties of $e$ and its endpoints.
\begin{equation}
\label{equation:update}
{
FrontierPool \underset{}{\longleftarrow}  \textbf{\textsc{Update}}(e) 
}
\end{equation}
}
\vspace{-.1in}

\subsection{Case Study}
\label{sect:arch:case}

{\gsag} can support all graph sampling and random walk algorithms introduced in Section \ref{sect:background:samplerw}.
Fig.~\ref{fig:api_example} exhibits how to use {\gsag} to implement two popular algorithms: Node2vec and \new{multi-dimensional random walk.}

Without loss of generality, we use the simplest example, i.e., \new{multi-dimensional random walk} to illustrate how {\gsag} works, as shown in Fig.~\ref{fig:layer_example}.
$FrontierSize$ and $NeighborSize$ are set as \new{3 and 1 respectively}.
{\textsc{VertexBias}} is based on the degree of vertices in the frontier pool in \new{multi-dimensional random walk}.
{\textsc{EdgeBias}} returns 1, resulting in the same transition probability for every neighbor.
{\textsc{Update}} always adds the currently sampled neighbor to the FrontierPool.

\begin{figure}[t]
	\floatbox[{\capbeside\thisfloatsetup{capbesideposition={right,top},capbesidewidth=3.8cm}}]{figure}[\FBwidth]
	{
		\caption{A \new{multi-dimensional random walk} example. Assuming \{$v_8$, $v_0$, $v_3$\} in FrontierPool$_t$, we use \textsc{VertexBias} to select $v_8$ as the sampling frontier at iteration $t$. Based on \textsc{EdgeBias} in Fig.~\ref{fig:api_example}(d), we select $v_7$, and put it in sampled edges array. According to \textsc{Update}, {\gsag} further puts $v_7$ in FrontierPool$_{t+1}$ as \{$v_0$, $v_3$, $v_7$\}. Similar process continues until {\gsag} gathers adequate samples.\vspace{-.01in}
			}
	\label{fig:layer_example}
	}
	{
		\includegraphics[width=.95\linewidth]{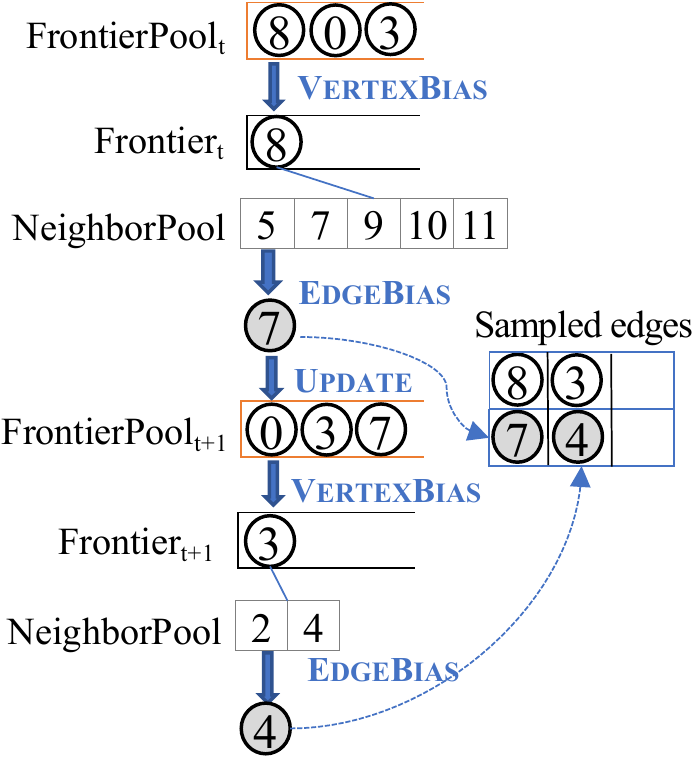}
	}
\end{figure}

\section{Optimizing GPU Sampling}
\label{sect:single}

Fig. \ref{fig:api}(b) has shown the overall algorithm of \gsag.
In this section, we discuss how to implement this algorithm efficiently on GPUs.
We will discuss our general strategies to parallelize the \textsc{Select} function on GPUs (Section \ref{sect:single:ves}) and how to address the conflict when multiple GPU threads select the same vertex (Section \ref{sect:single:bitmap}).

\subsection{{Warp-Centric Parallel Selection}}
\label{sect:single:ves}

The core part of the \gsag~algorithm is to {\em select} a subset of vertices from a pool (lines 4 and 6 in Fig. \ref{fig:api}(b)).
As discussed in Section~\ref{sect:background:ns}, several algorithms have been proposed in this regard.
In this paper, we adopt inverse transform sampling \cite{olver2013fast} for GPU vertex selection, because 1) it allows to calculate transition probabilities with flexible and dynamic biases, and 2) it shows more regular control flow which is friendly to GPU execution.
Fig. \ref{fig:select} illustrates the \textsc{Select} algorithm using inverse transform sampling.
We aim to have an efficient GPU implementation of it.





\textbf{Inter-warp Parallelism.}
{Each thread warp, no matter intra or inter thread blocks, is assigned to sample a vertex in \cready{$FrontierPool$}}.
To fully saturate GPU resources, thousands of \cready{candidate vertices needs to be sampled concurrently.}
There are two sources of them.
First of all, many sampling algorithms naturally \cready{sample all vertices in $FrontierPool$ concurrently.}
For instance, \cready{neighbor sampling allows all vertices in $FrontierPool$ to be sampled concurrently and} requires a separate \cready{$NeighborPool$} for each vertex in the \cready{$FrontierPool$.}

\begin{figure}[t]
	\centering
	\includegraphics[width=0.7\linewidth]{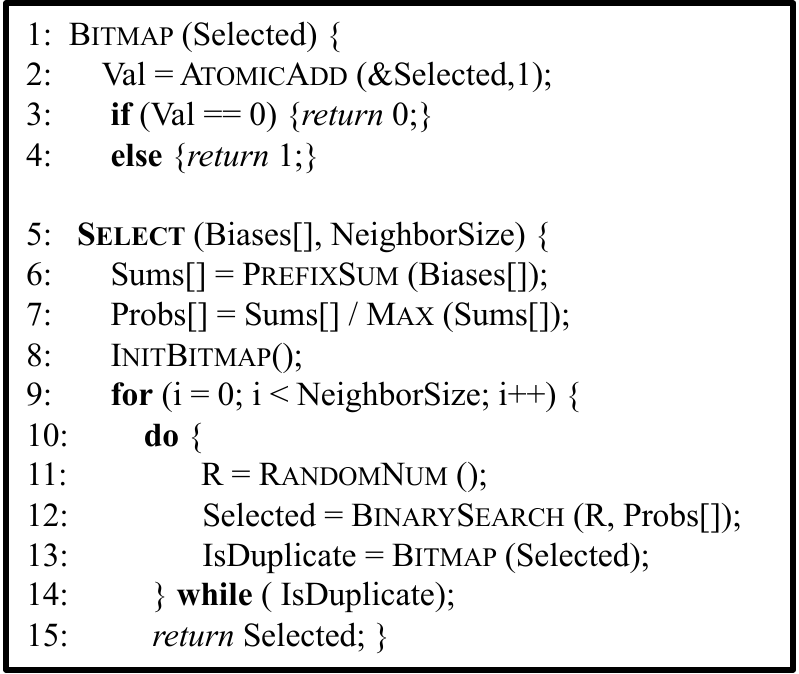} 
	\vspace{-.05in}
	\caption{\new{The unoptimized implementation of \textsc{Select} function}.
	\vspace{-.02in}}
	\label{fig:select}
\end{figure}

\begin{figure*}[t]
	\centering
	\includegraphics[width=0.99\linewidth]{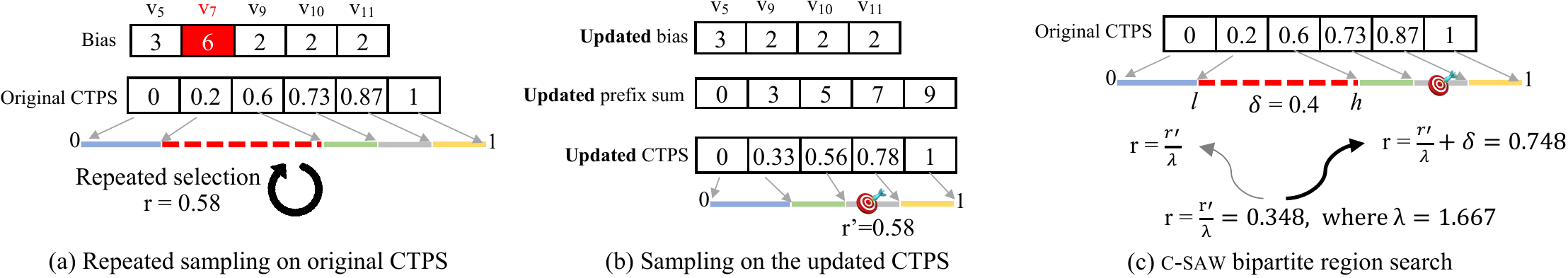}
	\caption{\new{Assuming $v_7$ is already selected (dotted line in CTPS): (a) naive repeated sampling on the original CTPS, (b) updated sampling on the recalculated CTPS, and (c) our bipartite region search approach.}
	\vspace{-.02in}
	}
	\label{fig:brs}
\end{figure*}

\cready{Second}, most sampling applications including {Graph Convolutional Network} (GCN)~\cite{kipf2016semi} , Deepwalk, Node2vec, and {Personalized PageRank} (PPR)~\cite{lofgren2014fast}, need to launch many instances of sampling either from the same seeds or different seeds. \new{Here, an \textit{instance} generates one sampled graph from the original graph. Particularly, for all algorithms except multi-dimensional random walk, an instance starts with one source vertex. For multi-dimensional random walk, an instance has multiple source vertices, which collectively generate one sampled graph.}
Applications like GCN require multiple sample instances for training the model~\cite{gao2018large,zeng2019graphsaint,chen2018fastgcn}, while Deepwalk, Node2vec, and PPR require multi-source random walk to either generate vertex embeddings or estimate PPR~\cite{alamgir2010multi, perozzi2014deepwalk, grover2016node2vec}.
With thousands of concurrent instances, {\gsag} is able to leverage the full computing power of GPU.
Since the inter-warp parallelism is straightforward to implement, we focus on exploiting the intra-warp parallelism for {\gsag}.

\textbf{Intra-warp Parallelism.}
A thread warp is used to execute one instance of \textsc{Select} on a pool of vertices.
An obvious alternative is to use a thread block. Most real world graphs follow power-law degree distribution, i.e., the majority of the vertices in the graph have very few edges. Using a thread block for a neighbor pool will fail to saturate the resource.
Our evaluation shows that using thread warps achieves $\sim 2 \times$ speedup compared with using thread blocks.
Thus we choose to use thread warps to exploit the parallelism within \textsc{Select}.

As shown in Fig. \ref{fig:select}, first, \textsc{Select} calculates the prefix sum of the biases of all vertices (line {6}).
Fortunately, parallel prefix sum is a well-studied area on GPUs.
In this paper, we adopt the Kogge-Stone algorithm\cready{~\cite{merrill2009parallel}} which presents superior performance for the prefix sum of warp-level where all threads execute in lock-step.
The normalization of prefix sums (line {7}) can be naturally parallelized by distributing the division of different array elements across threads.

To parallelize the vertex selection loop (line {10-14}), 
{\gsag} dedicates one thread for each vertex selection to maximize the parallelism.
For each loop iteration, a random number is generated to select one vertex, as introduced in Section \ref{sect:background:ns}.
However, this creates a crucial challenge that different threads may select the same vertex, i.e., {\em selection collision}.





\subsection{Migrating Selection Collision}
\label{sect:single:bitmap}

To migrate the aforementioned {\em selection collision},
we propose two interesting solutions: bipartite region search, and bitmap based collision detection. Before introducing our new design, we first discuss naive solutions.  
\vspace{0.05in} 

\textbf{Naive Solutions.}
A naive solution is to have a do-while loop (line {10-14} in Fig. \ref{fig:select}) to re-select another one until success, i.e., {\em repeated sampling}.
However, many iterations may be needed to make a successful selection.
As shown in Fig.~\ref{fig:brs}(a), if the region of $v_7$ (i.e., 0.2 - 0.6 in CTPS) is already selected, our newly generated random number 0.58 will not lead to a successful selection.
In fact, our evaluation observes that this method suffers for scale-free graphs whose transition probability can be highly skewed, or when a large fraction of the candidates need to be selected i.e larger $NeighborSize$.

Another solution is to recalculate the CTPS by excluding the already selected vertices, i.e., {\em updated sampling}, such as Fig.~\ref{fig:brs}(b).
Then we can always pick unselected vertices by searching through the updated CTPS.
Particularly in Fig.~\ref{fig:brs}(b), we will perform another Kogge-Stone prefix-sum for the new bias array \{3, 2, 2, 2\} towards \{0, 3, 5, 7, 9\}. Consequently, the CTPS becomes \{0, 0.33, 0.56, 0.78, 1\}. Then, the random number $r=0.58$ selects $v_{10}$.
Recalculating prefix sum is, however, time consuming.

\vspace{0.05in} 
\textbf{Bipartite Region Search} inherits the advantages of both repeated and updated sampling, while avoiding their drawbacks.
That is, it {\em does not need the expensive CTPS update} compared with updated sampling, while greatly {\em improve the chance of successful selection} compared with repeated sampling.

Particularly, while updated sampling updates the CTPS without changing the random number as shown in Fig. \ref{fig:brs}(b), the key idea of bipartite region search is to adjust the random number $r$ so that the CTPS remains intact and can be reused.
Most importantly, bipartite region search warrants that its random number adjustment leads to the same selections as updated sampling.
Note, this method is called bipartite region search because when the random number selects an already selected vertex, bipartite region search searches either the right or the left side of the already selected region in CTPS. 
Below, we discuss this adjustment.





\vspace{0.05in}
{\footnotesize
\centering
\noindent\fbox{%
\parbox{0.95\linewidth}{%
\noindent \circled{1} Generate a random number $r'$ $(0 \leq r' < 1)$.


\noindent \circled{2} Use $r'$ to select a vertex in CTPS. If the vertex has not been selected, done. Otherwise, the region that $r'$ falls into corresponds to a pre-selected vertex. Assume the boundary of this region in CTPS is $(l, h)$. Go to \circled{3}.

\noindent {\circled{3} Let $\lambda = 1/(1 - (h - l))$, $\delta = h - l$ and update $r$ to $r'/\lambda$. If $r < l$, select $(0, l)$ and go to \circled{4}.
Otherwise select $(h, 1)$ and go to \circled{5}.}

\noindent {\circled{4} Use the updated $r$ to search in $(0, l)$. If updated $r$ falls in another selected region, go to \circled{1}. Otherwise done.}

\noindent {\circled{5} Further update $r$ to $r + \delta$ and search in $(h, 1)$. If updated $r$ falls in another selected region, go to \circled{1}. Otherwise done.}
}
}
\par}
\vspace{0.05in}

{Fig.~\ref{fig:brs}(c) explains how bipartite region search works for the same example in Fig.~\ref{fig:brs}(b). 
Assuming we get a random number $r'=0.58$, it corresponds to $v_7$ in the original CTPS. 
Since $v_7$ is already selected, bipartite region search will adjust this random number to 0.348 in \circled{3}.
Since the updated $r = 0.348 > l = 0.2$, bipartite region search selects $(0.6, 1)$ to explore. Consequently in \circled{5}, we further add $\delta = 0.4$ to $r$ which leads to $r = 0.748$.
0.748 corresponds to $v_{10}$, and thus results in a successful selection. 
\cready{\textit{It is important to note that this selection is identical as updated sampling in Fig.~\ref{fig:brs}(b).}} }

\vspace{0.05in} 
\textbf{Proof of Bipartite Region Search.}
We will prove the soundness of bipartite region search mathematically, in the scenario when one and only one vertex has been pre-selected.



\begin{theorem}\label{thm:brs}
Assuming $v_k$'s probability region is $(F_k, F_{k+1})$ in the original CTPS. Remind the definition of $F$ in Section \ref{sect:background:ns}. Let $v_s$ be the pre-selected vertex, and $F'_k$ be the probability in the updated CTPS. $l = F_k$, $h = F_{k+1}$, $\lambda = \frac{1}{1 - (h - l)}$ and $\delta = h - l$, we prove that:

{\footnotesize
\begin{equation}\label{eq:brs}
      F'_{k} =
    \begin{cases}
      \lambda\cdot F_k; & \text{$k<s,$}\\
      \lambda\cdot (F_k - \delta); & \text{otherwise.}
    \end{cases}       
\end{equation}
}
\end{theorem}

\begin{proof}

Adopting Equation~\ref{equ:tp}, we get \new{$F_k = \frac{\sum_{i=1}^{k-1}b_{i}}{\sum_{i=1}^{n}b_{i}}$}. \new{Denoting $\mathbb{F}=\sum_{i=1}^{s-1}b_{i} + \sum_{i=s+1}^{n}b_{i}$}, Theorem~\ref{thm:tp} leads to: 

\new{\footnotesize
\begin{equation}
      F'_{k} =
    \begin{cases}
      \frac{\sum_{i=1}^{k-1}b_{i}}{\mathbb{F}}; & \text{$k<s,$}\\
      \frac{\sum_{i=1}^{s-1}b_{i} + \sum_{i=s+1}^{k-1}b_{i}}{\mathbb{F}}; & \text{otherwise.}
    \end{cases}       
\end{equation}
}

\noindent When $k<s$,

\new{\footnotesize
\begin{align}
F'_k&=\frac{\sum_{i=1}^{k-1}b_{i}}{\mathbb{F}}
    =\frac{\sum_{i=1}^{k-1}b_{i}}{\sum_{i=1}^{n}b_{i}}\cdot \frac{\sum_{i=1}^{n}b_{i}}{\mathbb{F}}
    =F_k\cdot\frac{\sum_{i=1}^{n}b_{i}}{\mathbb{F}}.
\end{align}
}

\noindent Since \new{$\frac{\sum_{i=1}^{n}b_{i}}{\mathbb{F}} = \frac{1}{1 - (h - l)} = \lambda$}, we prove $F'_k = \lambda\cdot F_k$.
When $k > s$,

\new{\footnotesize
\begin{equation}
\begin{aligned}
F'_k&=\frac{\sum_{i=1}^{s-1}b_{i} + \sum_{i=s+1}^{k-1}b_{i}}{\mathbb{F}}
    =\frac{\sum_{i=1}^{s-1}b_{i} + \sum_{i=s+1}^{k-1}b_{i}}{\sum_{i=1}^{n}b_{i}}\cdot \frac{\sum_{i=1}^{n}b_{i}}{\mathbb{F}}\\
    &=\frac{\sum_{i=1}^{s-1}b_{i} + \sum_{i=s+1}^{k-1}b_{i}}{\sum_{i=1}^{n}b_{i}}\cdot\lambda
    =\frac{\sum_{i=1}^{k-1}b_{i} - b_s}{\sum_{i=1}^{n}b_{i}}\cdot\lambda\\
    &=(\frac{\sum_{i=1}^{k-1}b_{i}}{\sum_{i=1}^{n}b_{i}} - \frac{b_s}{\sum_{i=1}^{n}b_{i}})\cdot\lambda
    =(F_k - \frac{b_s}{\sum_{i=1}^{n}b_{i}})\cdot\lambda.
\label{equ:factor}
\end{aligned}
\end{equation}
}

\noindent Since \new{$\frac{b_s}{\sum_{i=1}^{n}b_{i}} = {h-l} = \delta$}, we obtain $F'_k = \lambda\cdot (F_k -\delta)$.
\end{proof}

Theorem~\ref{thm:brs} states that one can adjust the probabilities from the original CTPS to derive the updated CTPS. 
Reversing the transformation direction, we further obtain: 

{\footnotesize
\begin{equation}\label{eq:brs_random}
      F_{k} =
    \begin{cases}
      \frac{F'_k}{\lambda}; & \text{$k<s,$}\\
      \frac{F'_k}{\lambda} + \delta; & \text{otherwise.}
    \end{cases}       
\end{equation}
}

Since $r'$ is the random number for the updated CTPS, we can substitute $F'_k$ with $r'$ in Equation~\ref{eq:brs_random} to derive the corresponding $r$ in the original CTPS that falls right at the region boundaries of original CTPS, e.g., \{0, 0.33, 0.56, 0.78, 1\} in Fig.~\ref{fig:brs}(b) fall right at \{0, 0.2, 0.73, 0.87, 1\} in Fig.~\ref{fig:brs}(c). Further, since $F_k$ is a strictly monotonic function of $F'_k$, it is clear that if $r'$ falls between the region boundaries of the updated CTPS, the derived $r$ will also do so in the original CTPS.
This ensures bipartite region search will make identical selection as if the CTPS is updated.
It is also provable that statistically, the selection probability of our algorithm is the same as the desired transition probability in more complicated scenarios where multiple vertices have been pre-selected.
\vspace{0.05in} 
\textbf{Strided Bitmap for Collision Detection.}
Bipartite region search requires a collision detection mechanism.
We introduce a per vertex bitmap to detect selection collision (line \new{13} in Fig. \ref{fig:select}).
For every candidate vertex, there is a unique bit in the bitmap to indicate whether it has been selected.
The bitmap is shared by all threads of a warp.
After each thread selects a vertex, we perform an atomic compare-and-swap operation to the corresponding bit in the bitmap. 
If the bit is 0, which means no other threads have picked this vertex, we set it to 1.

Since GPUs do not have variables that support bit-wise atomic operations currently, we may use either 8-bit or 32-bit integer variables for bitmap representation, where each bit corresponds to one vertex.
As using 32-bit variables results in more conflicts when updating multiple bits within the same variable, we choose 8-bit variables instead.

To resolve the atomic contentions, we propose to use {\em strided} bitmaps, inspired by the set-associative cache organization \cite{jouppi1990improving}.
A strided bitmap scatters the bits of adjacent vertices across different 8-bit variables, as shown in Fig. \ref{fig:bitmap}. Instead of using the first five bits of the same 8-bit variable to indicate the status of all vertices in the contiguous bitmap, the strided bitmap spreads them into two variables to reduce conflicts.

\begin{figure}[t]
    \centering
    \includegraphics[width=0.95\linewidth]{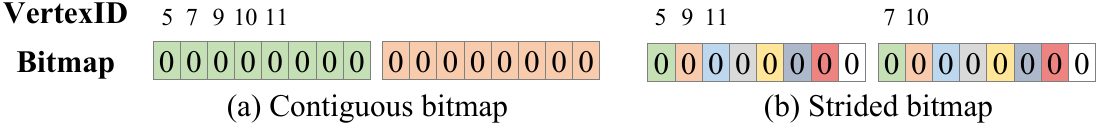}
    \caption{Sampling the neighbors of $v_8$ in Fig. \ref{fig:example}(a), under: (a) contiguous bitmap and (b) strided bitmaps.
    \vspace{-.1in}
    }
    \label{fig:bitmap}
\end{figure}

\new{
\vspace{0.05in} 
\textbf{Data Structures.}
{\gsag} employs three major data structures: frontier queues, per-warp bitmap, and per-warp CTPS.
All these data structures are allocated in the GPU global memory before sampling starts.
A frontier queue is a structure of three arrays, $VertexID$, $InstanceID$, and $CurrDepth$ to keep track of the sampling process.
Till now, all threads share one frontier queue, with a few exceptions that will be introduced in Section \ref{sect:outmem}.
Per-warp bitmaps and CTPSs are stored as arrays and get reused across the entire sampling process.
They are also located in global memory.
}

\section{Out-of-memory \& \new{Multi-GPU} {\gsag}}
\label{sect:outmem}


Thanks to sampling and random walk which lift important obstacles for out-of-memory computation, 
that is, 
they need neither the entire graph nor synchronization during computation. 
This section takes advantage of this opportunity to enable fast out-of-memory \new{and multi-GPU} {\gsag}.

\subsection{Graph Partition}

{\gsag} \new{partitions the graph by simply assigning a contiguous and equal range of vertices and all their neighbor lists to one partition.} We adopt this method instead of advanced topology-aware partition (e.g., \textsc{METIS}~\cite{karypis1995metis, karypis1998fast,guattery1995performance}) and 2-D partition~\cite{boman2013scalable}, for \new{three} reasons. First and foremost, sampling and random walk require all the edges of a vertex be present in order to compute the transition probability.
Splitting the neighbor list of any vertex, which is the case in 2-D partition, would introduce fine-grained communication between partitions, that largely hampers the performance. 
Second, topology-aware partition would require extremely long preprocessing time, as well as yield discontinued vertex ranges which often lead to more overhead than benefit. 
\new{Third, this simple partitioning method allows {\gsag} to decide \cready{which} partition a vertex belongs to in constant time that is important for fast bulk asynchronous sampling (Fig.~\ref{fig:streaming}).} 

\subsection{Workload-Aware Partition Scheduling}
\label{sect:outmem:design}

Since multiple sampling instances are independent of each other, this dimension of flexibility grants {\gsag} the freedom of dynamically scheduling various partitions based upon the workload from both graph partitions and workers (such as GPU kernels and devices). 
\vspace{0.05in}

\textbf{Workload-Aware Partition Scheduling.}
{\gsag} tracks the number of frontier vertices that falls into each partition to determine which partition will offer more workload (\circled{1} in Fig. \ref{fig:streaming}). \new{We refer them as active vertices. Based upon the count, we also allocate thread blocks to each GPU kernel with thread block based workload balancing described in next paragraph.}
Subsequently, the partitions that contain more workload are transferred to the GPU earlier and sampled first (\new{\circled{2}} in Fig. \ref{fig:streaming}).
Non-blocking {cudaMemcpyAsync} is used to copy partitions to the GPU memory asynchronously.
{\gsag} samples this partition until it has no active vertices. \new{
Note that, {\gsag} stores frontier queues from all partitions in the GPU memory. It allows a partition to insert new vertices to its frontier queue, as well as the frontier queues of other partitions to enable communications.}
The actively sampled partition is only released from the GPU memory when its frontier queue is empty.
The reason is that partitions with more active vertices often insert more neighbors in its own frontier queue, which further leads to more workloads.
As a result, this design can reduce the number of partitions transferred from CPU to GPU.

When it comes to computation, we dedicate one GPU kernel to one active partition along with a CUDA stream, in order to overlap the data transfer and sampling of different active partitions.
After parallel partition sampling finishes, we count the vertex number in each frontier queue to decide which partitions should be transferred to GPU for sampling next (\circled{3} in Fig. \ref{fig:streaming}).
The entire sampling is complete when there are no active vertices in all partitions.
\vspace{0.1in}

\textbf{Thread Block based Workload Balancing.}
Depending upon the properties of graphs and sample seeds, frontiers are likely not equally distributed across partitions.
As a result, the sampling and data transfer time are not the same as well.
Since the straggler kernel determines the overall execution time, it is ideal to balance the workload across kernels. 
Consequently, we implicitly partition the GPU resources by controlling the thread block number of different kernels.

\textbf{Example.}
Fig.~\ref{fig:streaming} shows an example of out-of-memory sampling.
Here, we assume three graph partitions (i.e., P$_1$, P$_2$, P$_3$) for \new{the same graph in Fig.~\ref{fig:example}(a)}, two GPU kernels (i.e., Kernel$_1$ and Kernel$_2$), and the GPU memory can contain two active partitions.
\new{If we start sampling from vertices \{0, 2, 8\}}, P$_1$, P$_2$, and P$_3$\ will have \new{2, 0, and 1} active vertices initially.
Hence, kernel K$_1$ is assigned to work on \new{P$_1$} and kernel K$_2$ for P$_3$. To balance the workload, the ratio of thread block numbers assigned to K$_1$ and K$_2$ is set to \new{2:1}.
\new{Assuming vertices 0, 2, and 8 pick 7, 3, and 5, respectively, the frontier queues for P$_1$, P$_2$ and P$_3$ become \{3\}, \{7, 5\} and \{$\phi$\} as shown in bottom right of Fig.~\ref{fig:streaming}. Subsequently, K$_2$ exits because P$_3$'s frontier queue is empty, while K$_1$ continues sampling 3 and puts 4 into the frontier queue of P$_2$.
Then, K$_1$ also exits and leaves \{7, 5, 4\} in the frontier queue of P$_2$ to be scheduled next.}

\begin{figure}[t]
    \centering
    \includegraphics[width=.95\linewidth]{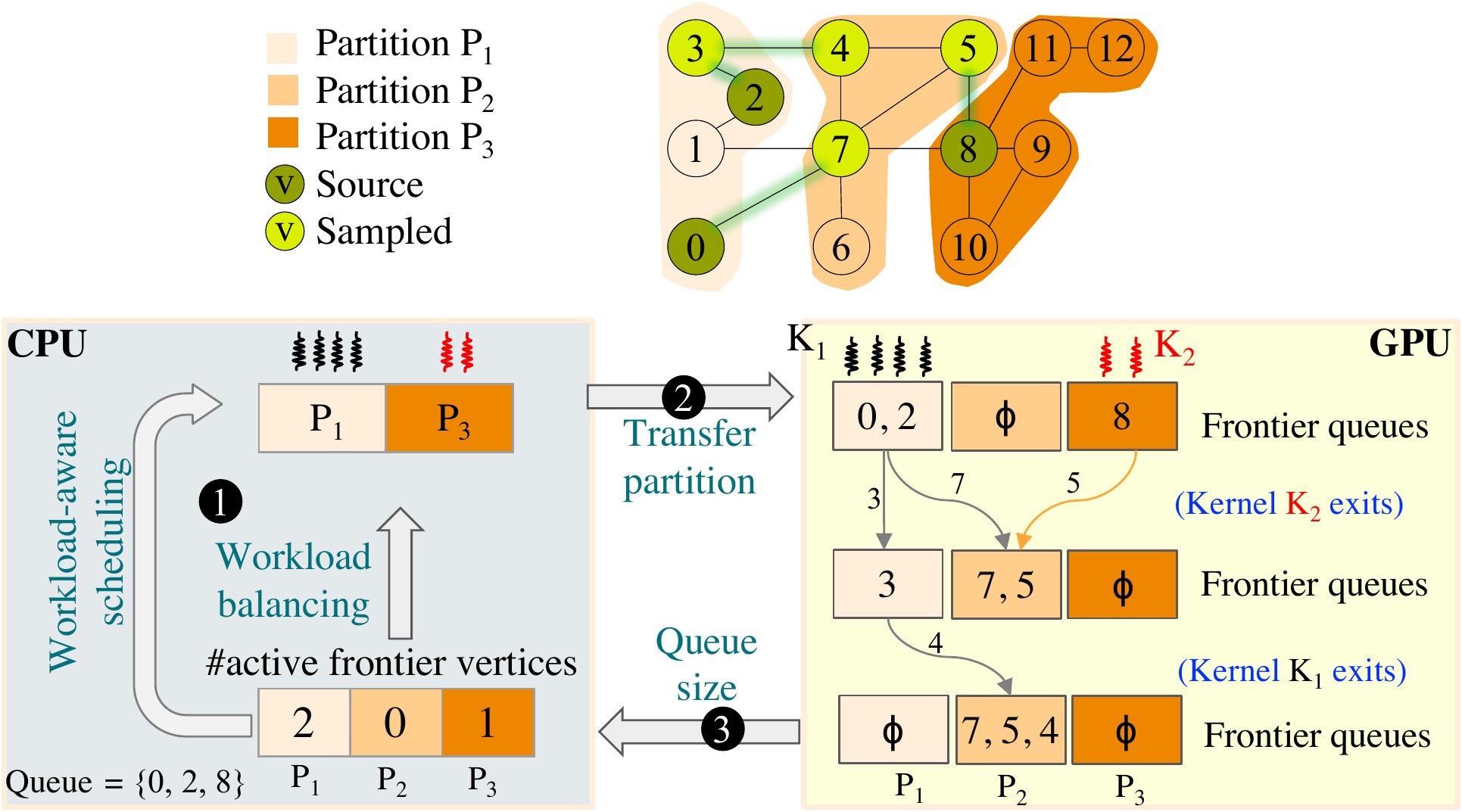}
    \caption{Workload-aware scheduling of graph partition. \new{The upper part shows the toy graph and its partition. We start sampling within partitions 1 and 3. The lower part shows an example for out-of-memory sampling. For simplicity, we hide InstanceID and CurrDepth from the frontier queue. }
    \vspace{-.1in}
    }
    \label{fig:streaming}
\end{figure}

\vspace{0.05in} 
\new{\textbf{Correctness.}
The out-of-order nature of the workload-aware partition scheduling does not impact the correctness of {\gsag}.
With out-of-order scheduling, the sampling of one instance is not in the breath-first order as in the in-order case.
The actual sampling order can be considered as a combination of breath-first and depth-first orders.
However, since we keep track of the depth of sampled vertices to prevent an instance from reaching beyond the desired depth, the sampling result will be the same as if it is done in the breath first order.
}

\subsection{Batched Multi-Instance Sampling}
\label{sect:multi:balance}

In the out-of-memory setting, {\gsag} introduces {\em batched multi-instance sampling}, which \textit{concurrently} samples multiple instances, to combat the expensive data transferring cost.

Batched sampling is implemented by combining the active vertices of various concurrently sampling instances into a single frontier queue for each partition.
Along with the queue, we need to keep two extra metadata for each vertex, i.e., $InstanceID$ \new{and $CurrDepth$, which tracks the instance that a vertex belongs to and stores the current depth of that instance respectively.}
During sampling, a thread warp in the kernel can work on any vertex in the queue, no matter whether they are from the same or different instances.
After it finishes selecting vertices (line 6 in Fig. \ref{fig:api}(b)), $InstanceID$ is used to find the corresponding frontier pool and sampled graph to update (line 7-8).
Note that there may exist multiple copies of the same vertex in the queue, because a common vertex can be sampled by multiple instances.

\begin{figure*}[t]
    \centering
    \subfloat[{\gsag} vs. KnightKing on biased random walk.]{
        \includegraphics[width=.495\linewidth]{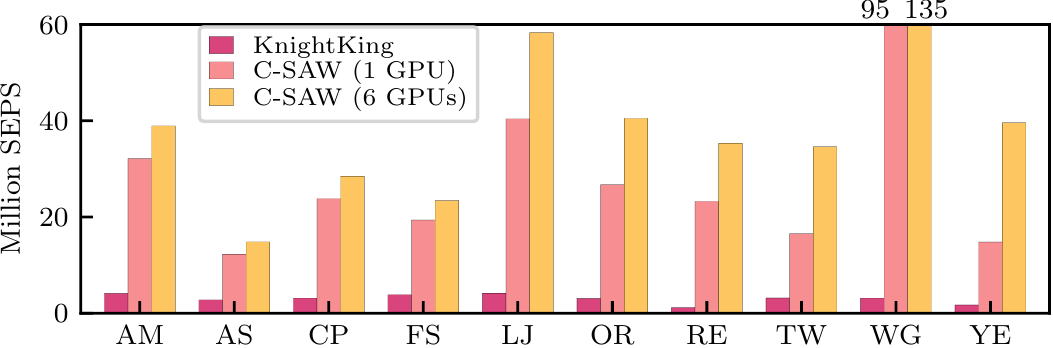}
    }%
    \subfloat[{\gsag} vs. GraphSAINT on multi-dimensional random walk.]{
        \includegraphics[width=.475\linewidth]{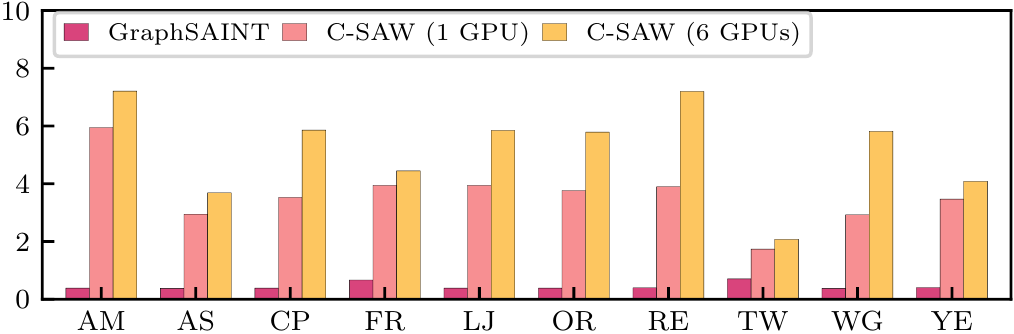}
    }
    \caption{\new{{\gsag} vs. the state-of-the-art in million sampled edges per second with 1 GPU and 6 GPUs (higher is better).} 
    }
    \label{fig:stateofart}
\end{figure*}

Batched sampling can also balance the workload across sampling instances. Otherwise, if we sample various instances separately,
since many real-world graphs hold highly skewed degree distributions, some instances may encounter higher degree vertices more often and thus more workloads. This will end up with skewed workload distributions.
Batched sampling solves this problem using a vertex-grained workload distribution, instead of instance-grained distribution.

\new{
\subsection{Multi-GPU {\gsag}}
As the number of sources continues to grow, the workload will saturate one GPU and go beyond. In this context, scaling {\gsag} to multiple GPUs would help accelerate the sampling performance. 
Since various sampling instances are independent from each other, {\gsag} simply divides all the sampling instances into several disjoint groups, each of which contains equal number of instances. Here, the number of disjoint groups is the same as the number of GPUs. Afterwards, each GPU will be responsible for one sampling group. During sampling, each GPU will perform the same tasks as shown in Fig.~\ref{fig:streaming} and no inter-GPU communication is required. 
}

\vspace{0.1in} 
\section{Evaluations}
\label{sect:experiment}

{\gsag} is {implemented with $\sim$4,000} lines of CUDA code and compiled by CUDA Toolkit 10.1.243 and g++ 7.4.0 
with optimization flag as -O3. 
We evaluate {\gsag} on the Summit supercomputer of Oak Ridge National Laboratory~\cite{ornl_summit}.
Each Summit node is equipped with 6 NVIDIA Tesla V100 GPUs, dual-socket 22-core POWER9 CPUs and 512 GB main memory. Particularly, each V100 GPU is equipped with 16GB device memory. 
For the random number generation, we use the cuRAND library~\cite{tian2009mersenne}.

\begin{table}[!h]
{\scriptsize
\begin{tabular}{|l|l|l|l|l|l|}
\hline
Dataset & Abbr. & \begin{tabular}[c]{@{}l@{}}Vertex \\ Count\end{tabular} & \begin{tabular}[c]{@{}l@{}}Edge \\ Count\end{tabular} & \begin{tabular}[c]{@{}l@{}}Avg.\\ degree\end{tabular} & \new{\begin{tabular}[c]{@{}l@{}}Size\\ (of CSR)\end{tabular}} \\ \hline
Amazon0601~\cite{snapnets}       & AM     & 0.4M   & 3.4M       & 8.39    &\new{59 MB} \\ \hline
As-skitter \cite{snapnets}       & AS      & 1.7M   & 11.1M    & 6.54   &\new{325 MB}\\ \hline
cit-Patents \cite{snapnets}      & CP     & 3.8M    & 16.5M   & 4.38   &\new{293 MB}\\ \hline
LiveJournal \cite{snapnets}     & LJ      & 4.8M   & 68.9M    & 14.23  & \new{1.1 GB}\\ \hline
Orkut \cite{snapnets}            & OR      & 3.1M   & 117.2M   & 38.14  & \new{1.8 GB}\\ \hline
Reddit \cite{zeng2019graphsaint,yelpreddit}           &RE       &0.2M        &11.6M     & 49.82 & \new{179 MB}\\ \hline
web-Google \cite{snapnets}       & WG      & 0.8M      & 5.1M    & 5.83  & \new{85 MB}\\ \hline
Yelp \cite{zeng2019graphsaint,yelpreddit}           &YE         &0.7M        &6.9M      & 9.73& \new{111 MB}\\ \hline \hline
Friendster \cite{snapnets}       & FR      & 65.6M   & 1.8M  & 27.53      &\new{29 GB} \\ \hline
Twitter \cite{konect:2017:twitter}          & TW      & 41.6M   & 1.5M  & 35.25  &\new{22 GB}\\ \hline
\end{tabular}
}
\caption{Details of evaluated graphs.}\label{Table-datasets} 
\end{table}

\textbf{Dataset.}
We use the graph datasets in Table~\ref{Table-datasets} to study {\gsag}. This dataset collection contains a wide range of applications, such as social networks (LJ, OR, FR and TW), forum discussion (RE and YE), online shopping (AM), citation networks (CP), computer routing (AS) and web page (WG).

\textbf{Metrics.}
Instead of Traversed Edges Per Second (TEPS) in classical graph analytics~\cite{liu2015enterprise,wang2016gunrock}, we introduce a new metric - Sampled Edges Per Second (SEPS) - to evaluate the performance of sampling and random walk. Formally, SEPS = $\frac{\#~\text{SampledEdges}}{\text{Time}}$. This metric is more suitable than TEPS to evaluate sampling and random walk because these algorithms might use different methods thus traverse a different number of edges but end up with the same number of sampled edges.
\new{Similar to previous work \cite{wang2016gunrock,liu2015enterprise}, the kernel execution time is used to compute SEPS, i.e., the time spent on generating the samples}, except for the out-of-memory case that also includes the time for transferring the partitions. Note, each reported result is an average of three runs with different sets of seeds.

\textbf{Test Setup.} 
Analogous to GraphSAINT~\cite{zeng2019graphsaint}, we generate \new{4,000} \new{instances} for random walk algorithms and 2000 instances for sampling algorithms.
For sampling, both the \new{\textit{NeighborSize}} (i.e., number of neighbors sampled from one frontier) and $Depth$ are 2 for analyzing the performance of {\gsag} except forest fire, which uses $P_{f}$= 0.7 to derive \new{\textit{NeighborSize}} as in~\cite{leskovec2006sampling}. For \new{biased random walk algorithm}, the length of the walk is \new{2,000}. \new{For multi-dimensional random walk, similar to GraphSAINT, we use \new{2,000} as the \textit{FrontierSize} for each instance.}

\subsection{{\gsag} vs. State-of-the-art}
\label{sect:experiment:in-mem}
\vspace{-0.03in}
First, we compare {\gsag} against the state-of-the-art frameworks, KnightKing and GraphSAINT. \new{Our profiling result shows that both GraphSAINT and KnightKing use multiple threads to perform the computation, where the \# threads = \# cores.}
Since KnightKing only supports random walk variations, we compare {\gsag} with KnightKing for \new{biased random walk}.
\new{GraphSAINT provides both Python and C++ implementations. We choose the C++ implementation~\cite{zeng2019accurate}
which exhibits better performance. \cready{The} C++ version only supports multi-dimensional random walk which is studied in Fig.~\ref{fig:stateofart}(b)}. 

As shown in Fig.~\ref{fig:stateofart}, {\gsag} presents superior performance over both projects. On average, {\gsag} is\new{ {10}$\times$ and {14.7}$\times$} faster than KnightKing with 1 GPU and 6 GPUs, respectively. Compared to GraphSAINT, {\gsag} is \new{{8.1}$\times$ and {11.5}$\times$ faster with 1 GPU and 6 GPUs respectively. Each instance of sampled graphs has 1,703 edges on average.} 
While {\gsag} outperforms both projects across all graphs, we generally observe better speedup on graphs with a lower average degree, such as, AM, CP and WG on KnightKing and AM on GraphSAINT. This is rooted from \new{1) the superior computing capability of GPU over CPU}, 2) {\gsag} is free of \cready{bulk synchronous parallelism (BSP)}~\cite{malewicz2010pregel}, which allows it to always have adequate computing tasks for sparse graphs, and 3) the unprecedented bandwidth of the V100 GPU over the {POWER9} CPU, i.e., 900 GB/s vs. {170 GB/s} \cite{ornl_summit}.
This underscores the need of GPU-based sampling and random walk.

\subsection{In-memory Optimization}

\begin{figure}[ht]
    \centering    
    \subfloat[\new{Biased neighbor sampling.}]{ 
        \hspace{-.05in}\includegraphics[width=.50\linewidth]{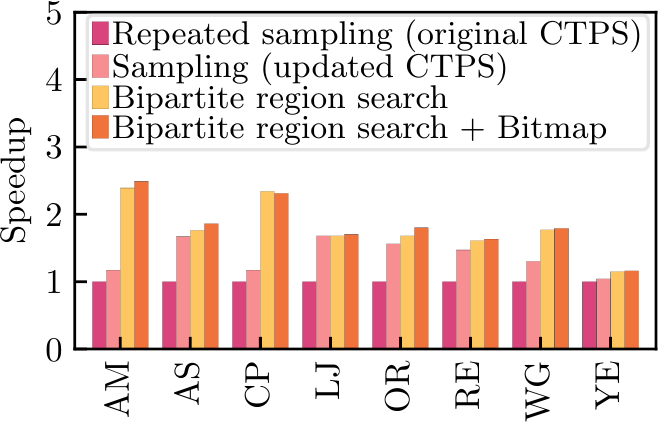}
    }
        \subfloat[\new{Forest fire sampling.}]{
        \hspace{-.05in}\includegraphics[width=.47\linewidth]{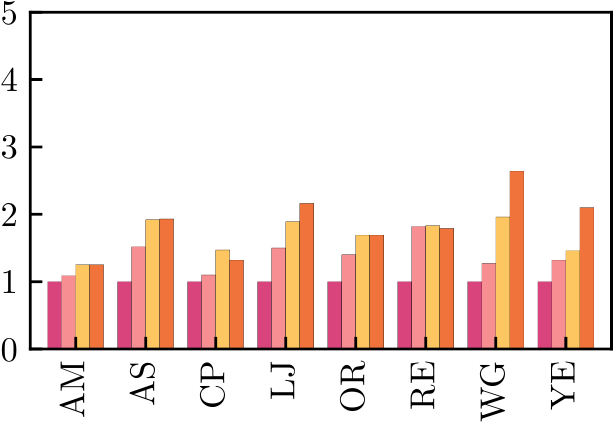}
    }\\
        \subfloat[\new{Layer sampling.}]{
        \hspace{-.05in}\includegraphics[width=.50\linewidth]{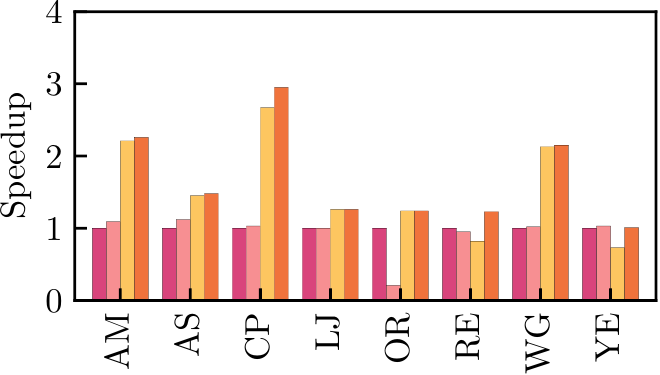}
    }
    \subfloat[\new{Unbiased neighbor sampling.}]{
        \hspace{-.05in}\includegraphics[width=.47\linewidth]{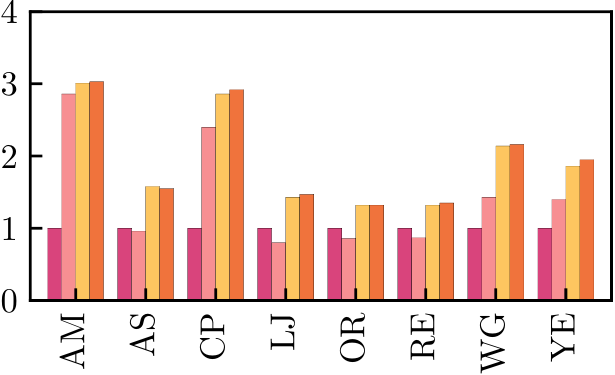}
    }
    \caption{\new{Performance impacts of in-memory optimizations for various sampling algorithms. \vspace{-0.1in}
    } 
    %
    }
    \label{fig:in_memory}
\end{figure}

\begin{figure}[ht]
    \centering
    \subfloat[\new{Biased neighbor sampling.}]{
        \hspace{-.01in}\includegraphics[width=.49\linewidth]{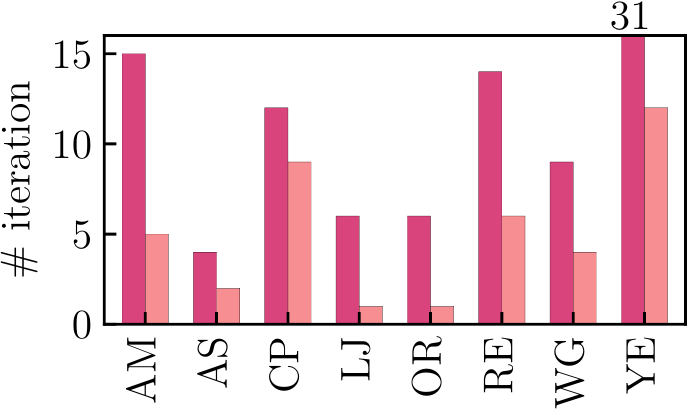}
    }
    \subfloat[\new{Forest fire sampling.}]{
        \hspace{-.01in}\includegraphics[width=.46\linewidth]{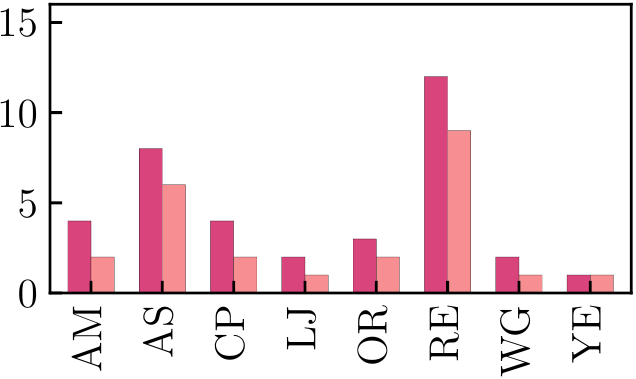}
    }
    \\ 
    \subfloat[\new{Layer sampling.}]{
        \hspace{-.01in}\includegraphics[width=.49\linewidth]{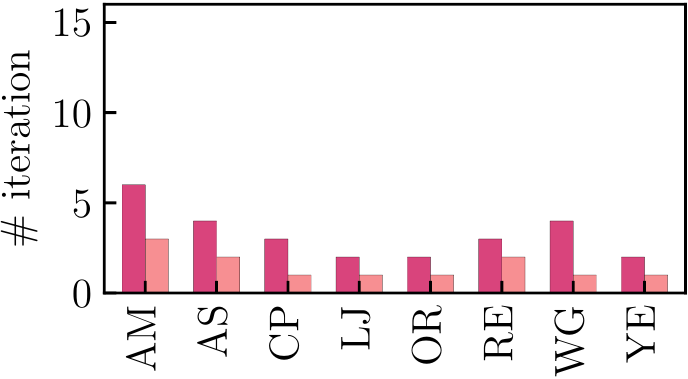}
    }
    \subfloat[\new{Unbiased neighbor sampling.}]{
        \hspace{-.01in}\includegraphics[width=.46\linewidth]{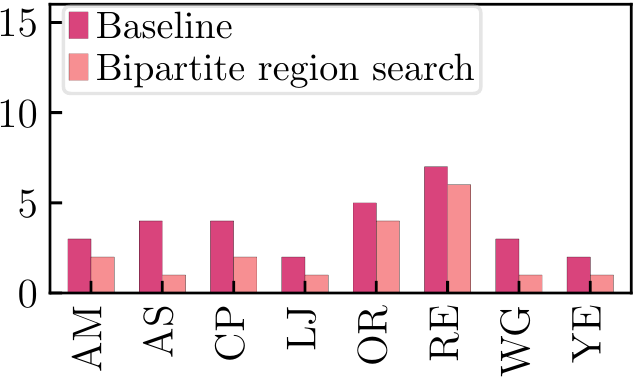}
    }
    \caption{Average \# iteration w/ and w/o \new{bipartite region search} for various algorithms.
    \vspace{-0.1in}
    }
    \label{fig:profile_bipartite}
\end{figure}

\begin{figure}[t]
    \centering
    
\subfloat[\new{Biased neighbor sampling.}]{
        \hspace{-.05in}\includegraphics[width=.50\linewidth]{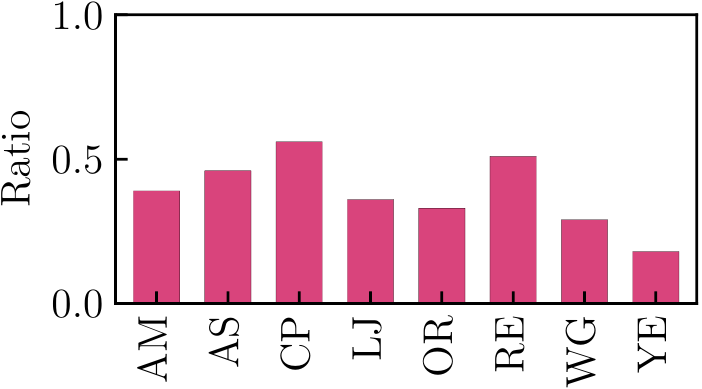}
    }
    \subfloat[\new{Forest fire sampling.}]{
        \hspace{-.05in}\includegraphics[width=.47\linewidth]{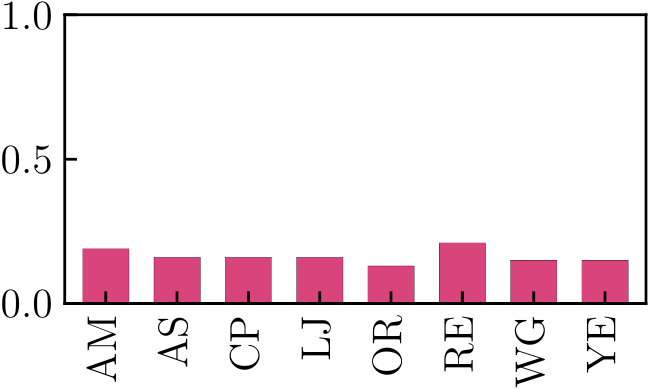}
    }
    \\ 
    \subfloat[\new{Layer sampling.} ]{
        \hspace{-.05in}\includegraphics[width=.50\linewidth]{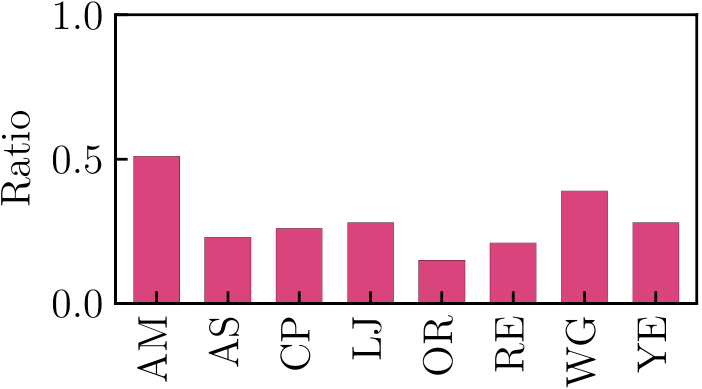}
    }
    \subfloat[\new{Uniased neighbor sampling.}]{
        \hspace{-.05in}\includegraphics[width=.47\linewidth]{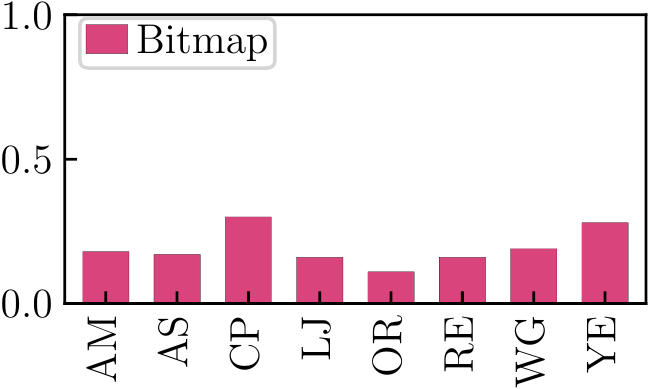}
    }
    \caption{\new{Total search reduction by bitmap for various algorithms.}
    }
    \label{fig:profile_bitmap}
\end{figure}

Fig.~\ref{fig:in_memory} studies the performance impacts of \new{bipartite region search} and bitmap optimizations over repeated sampling  (Fig.~\ref{fig:brs}(a)) and updated sampling (Fig.~\ref{fig:brs}(b)) across four applications, which include both biased and unbiased algorithms.
Repeated sampling is used as the performance baseline for comparison.
\new{FR and TW are not studied in this subsection because they exceed the GPU memory capacity.}
Particularly, \new{bipartite region search} introduces, on average, {1.7}$\times$, {1.4}$\times$, {1.7}$\times$ and {1.17}$\times$ speedup, on \new{biased neighbor sampling, forest fire sampling, layer sampling, and unbiased neighbor sampling} respectively. \new{Bipartite region search} presents better performance compared with both repeated sampling and updated sampling.
Bitmap further improves speedup to {1.8}$\times$, {1.5}$\times$, {1.8}$\times$, and {1.28}$\times$ on these four applications, respectively. The performance for AM, CP, and WG gleams the effectiveness of {\gsag}. With a lower average degree of vertices, they suffer from more selection collision. Using \new{bipartite region search}, we achieve better speedup by mitigating the collision. 

Fig.~\ref{fig:profile_bipartite} and~\ref{fig:profile_bitmap} further profile the effectiveness of our two optimizations. On average, \new{bipartite region search} reduces the average number of iterations to pick a neighbor by {5.0}$\times$, {1.5}$\times$, {1.8}$\times$, and {1.7}$\times$ for these four applications, respectively.
\new{
Here, \#~iterations refers to the trip count of do-while loop in Fig.~\ref{fig:select} (line 10-14), which represents the amount of computation used to select a vertex.
For analysis, we compare the average number of iterations for all sampled vertices, i.e., $\frac{Total~\#~\text{iterations of sampled vertices}}{\#~\text{sampled vertices}}$.}
We observe more \new{reduction on \#~iterations} for \new{biased neighbor sampling} than other algorithms as it has a higher selection collision chance and thus requires more iterations without \new{bipartite region search}.
With relatively larger neighbor pools, collision is less likely to happen in \new{layer sampling} which explains its lower benefits from bipartite region search.
Similarly, \new{unbiased neighbor sampling} and \new{forest fire sampling} incur less collision due to unbiased sampling.
Fig.~\ref{fig:profile_bitmap} shows the effectiveness of bitmap over the baseline which stores the sampled vertices in the GPU shared memory and performs a linear search to detect collision. \new{The ratio metric in Fig.~\ref{fig:profile_bitmap} compares the total number of searches performed by bitmap with that of baseline, i.e., $\text{Ratio} = \frac{\sum\#~\text{searches in bitmap}}{\sum\#~\text{searches in baseline}}$.} Compared to baseline, bitmap reduces the total searches by 
{{63}\%, {83}\%, {71}\%, and {81}\%} for these four applications, respectively. 
Despite of the significant search count reduction from bitmap, \cready{the} overhead of atomic operations refrains us from achieving speedups proportional with the search count reduction.

\subsection{Out-of-memory Optimization}

\begin{figure}[t]
 \centering

    \subfloat[\new{Biased neighbor sampling.} ]{
        \hspace{-.05in}\includegraphics[width=.50\linewidth]{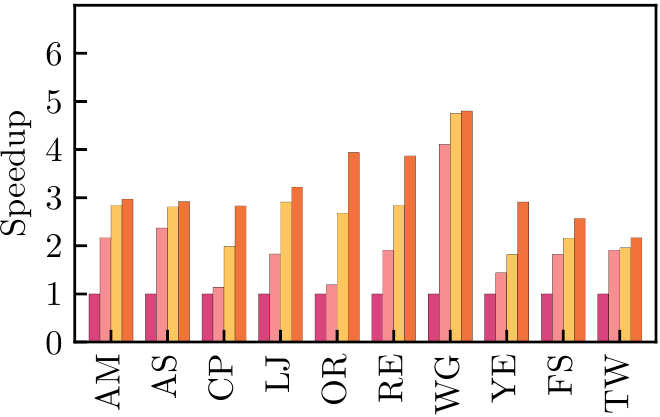}
    }
    \subfloat[\new{Biased random walk.} ]{
        \hspace{-.05in}\includegraphics[width=.47\linewidth]{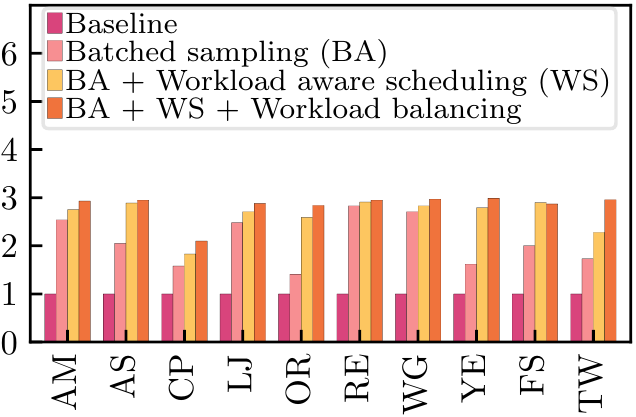}
    }\\ 
    \subfloat[\new{Forest fire sampling.} ]{
        \hspace{-.05in}\includegraphics[width=.50\linewidth]{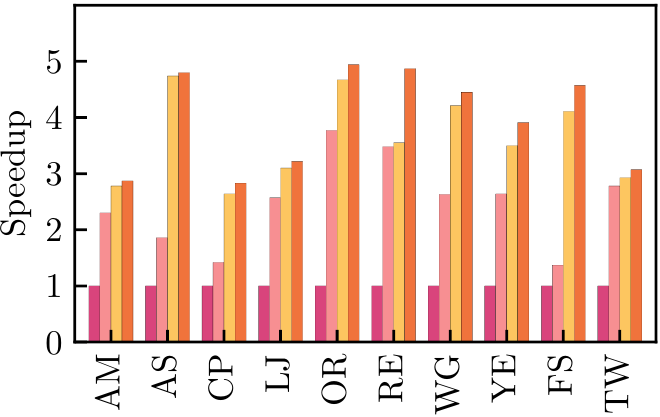}
    }
    \subfloat[\new{Unbiased neighbor sampling.} ]{
        \hspace{-.05in}\includegraphics[width=.47\linewidth]{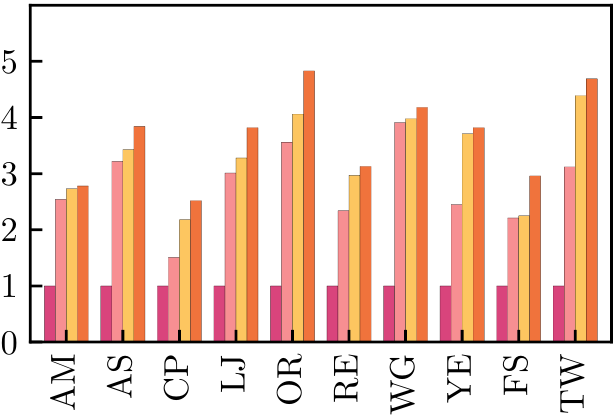}
    }
    \caption{Performance impacts of out-of-memory optimizations. Here,\new{ baseline implementation refers to partition transfer based on active partition without any optimization}.
    \vspace{-0.15in}
    }
    \label{fig:out_memory}
\end{figure}

\begin{figure}[t]

    \centering
        \subfloat[\new{Biased neighbor sampling.}]{
        \hspace{-.05in}\includegraphics[width=.50\linewidth]{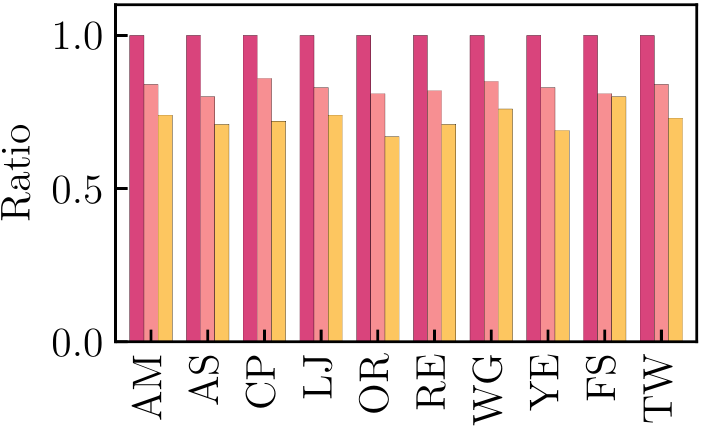}
    }
        \subfloat[\new{Biased random walk.}]{
        \hspace{-.05in}\includegraphics[width=.47\linewidth]{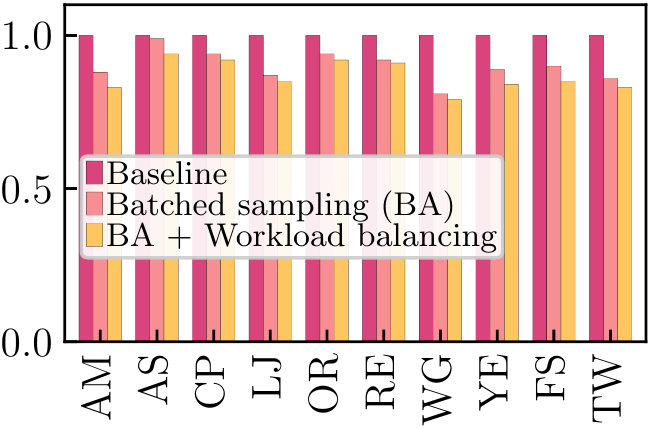}
    } \\
    \subfloat[\new{Forest fire sampling.} ]{
        \hspace{-.05in}\includegraphics[width=.50\linewidth]{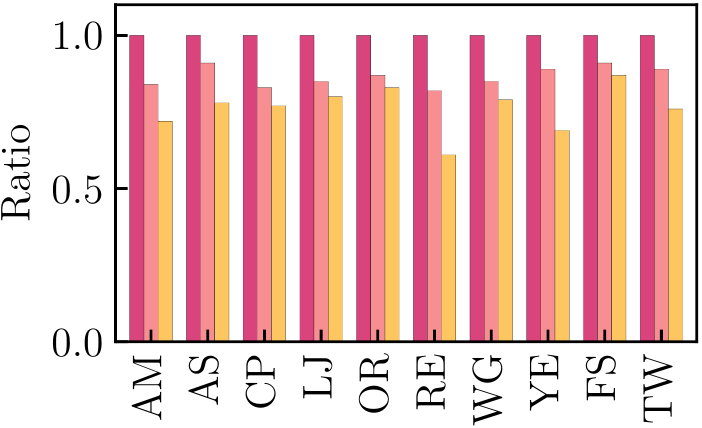}
    }
    \subfloat[\new{Unbiased neighbor sampling.}]{
        \hspace{-.05in}\includegraphics[width=.47\linewidth]{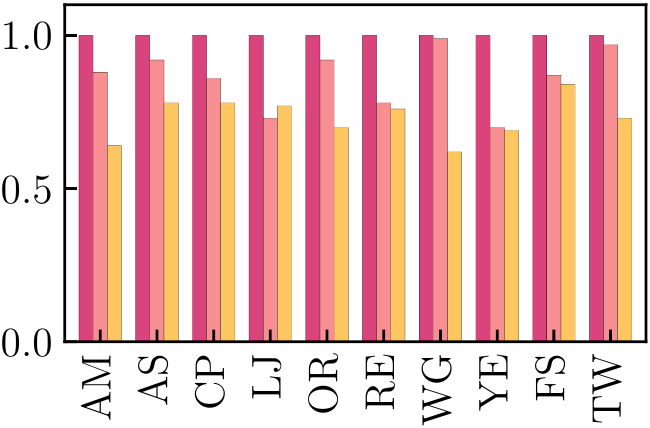}
    }
    \vspace{-.05in}
    \caption{Standard deviation of kernel time for multi-instance batching and workload-aware balancing in out-of-memory {\gsag} \new{(lower is better). Here, baseline represents even distribution of resources. \vspace{-0.15in}}
    }
    \label{fig:profile_batched}
\end{figure}

Fig.~\ref{fig:out_memory} presents the performance impacts of \new{multi-instance} batched sampling (BA), workload-aware scheduling (WS), and thread block based workload balancing (BAL) \new{on both large graphs and small graphs. For the sake of analysis, we pretend small graphs do not fit in GPU memory.} For the experimental analysis, we use 4 partitions for each graph and two CUDA streams.   
Assume the GPU memory can keep at most two partitions at the same time, for all graphs.
Particularly, batched sampling introduces, on average, {2.0}$\times$, {1.9}$\times$, {2.1}$\times$, and {2.7}$\times$ speedup, respectively on \new{biased neighbor sampling, biased random walk, forest fire sampling, and unbiased neighbor sampling}. Workload-aware scheduling further introduces {3.2}$\times$, {2.8}$\times$, {3.9}$\times$, and {3.3}$\times$ speedups on these four applications, respectively. Workload balancing gives, on average, {3.5}$\times$ speedup over all applications.

\begin{figure}[ht]

    \centering
    \subfloat[\new{Biased neighbor sampling.}]{
        \hspace{-.05in}\includegraphics[width=.50\linewidth]{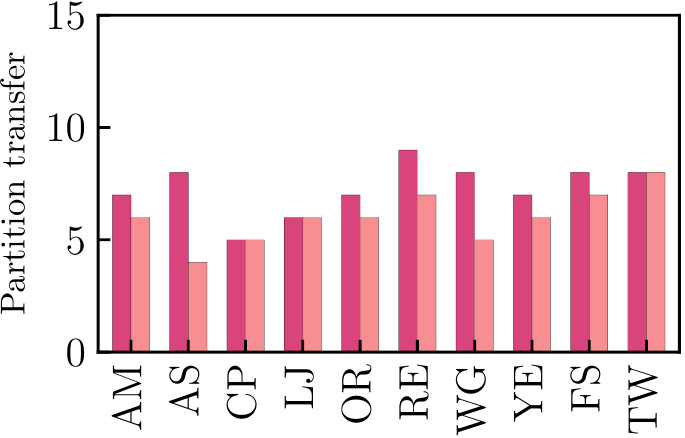}
    }
    \subfloat[\new{Biased random walk.}]{
        \hspace{-.05in}\includegraphics[width=.47\linewidth]{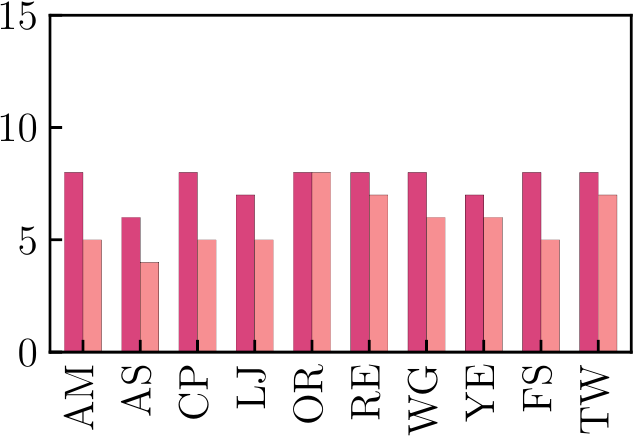}
    }
    \\ 
    \subfloat[\new{Forest fire sampling.}]{
        \hspace{-.05in}\includegraphics[width=.50\linewidth]{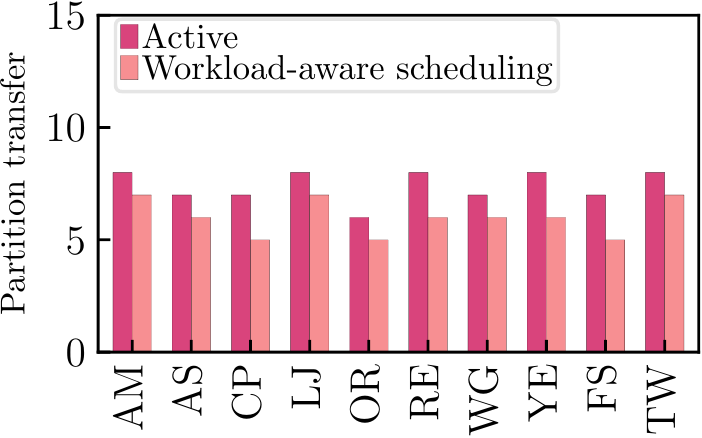}
    }
    \subfloat[\new{Unbiased neighbor sampling.}]{
        \hspace{-.05in}\includegraphics[width=.47\linewidth]{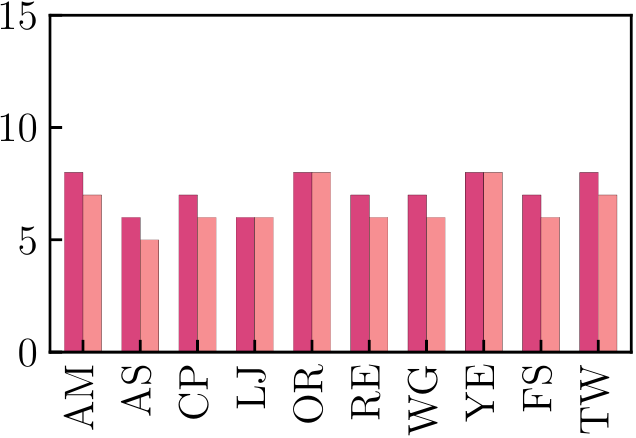}
    }
    \caption{Partition transfer counts for workload-aware scheduling \new{(lower is better)}.
    \vspace{-0.1in}
    }
    \label{fig:profile_degree}
\end{figure}

Fig.~\ref{fig:profile_batched} and~\ref{fig:profile_degree} reasons the effectiveness of two optimizations. \new{We use standard deviation to measure workload imbalance in runtime of two kernels for overall sampling.} On average, multi-instance batched sampling (BA) and thread block based workload balancing (BAL) reduce the average kernel time by {27}\%, {12}\%, {23}\%, and {26}\%, respectively on four applications.
As active vertices increase exponentially with depth during sampling, \new{biased neighbor sampling, forest fire sampling, and unbiased neighbor sampling} observe more reduction in kernel time than \new{biased random walk}. Workload-aware scheduling reduces the overall partition transfers by {1.2}$\times$, {1.3}$\times$, {1.2}$\times$, and {1.1}$\times$ on these four applications, respectively. Even with moderate decrease in partition transfers, we still achieve noticeable speedups.

\begin{figure}[ht]
    \vspace{-0.25cm}
    \centering    
    \subfloat[NeighborSize: 1 - 8.]{
        \hspace{-.05in}\includegraphics[width=0.95\linewidth]{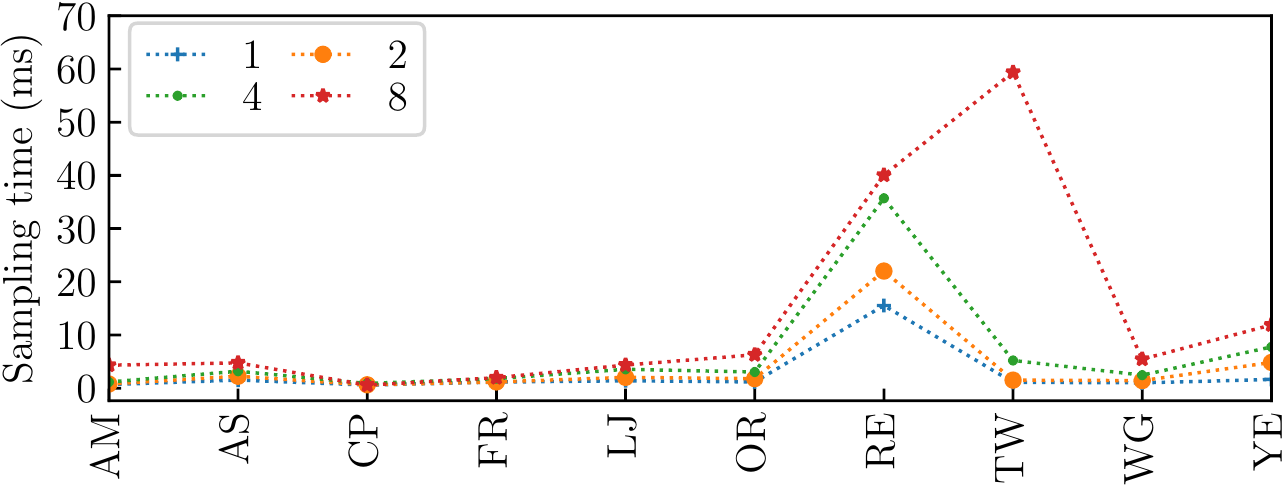}
    }
    \\
        \subfloat[\# instances: 2k - 16k.]{
        \hspace{-.05in}\includegraphics[width=0.95\linewidth]{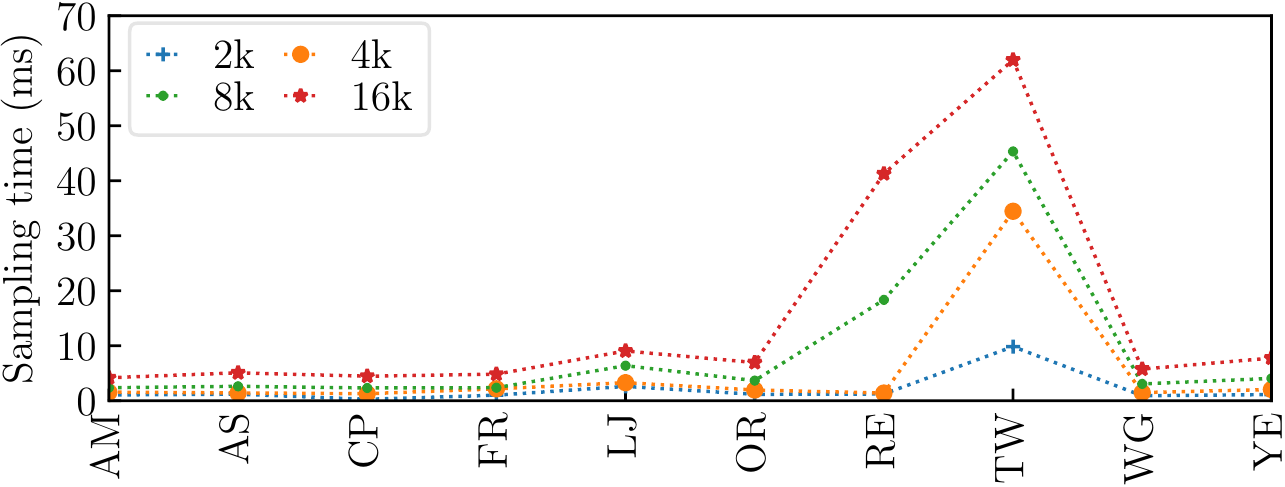}
    }
    \vspace{-0.1in}
    \caption{\new{Biased neighbor sampling with (a) \textit{NeighborSize} as 1, 2 4 and 8 and (b) \# instances as 2k, 4k, 8k and 16k.\vspace{-0.1in}}
    }
    \vspace{-0.05in}
    \label{fig:time}
\end{figure}

\vspace{0.05in}
\subsection{Studying NeighborSize and \#~Instances in {\gsag}}

Fig.~\ref{fig:time} reports the time consumption impacts of various \textit{NeighborSize} and \# instances. Here, we use Depth= 3 and 16k instances in Fig.~\ref{fig:time} (a) for extensive analysis. For Fig.~\ref{fig:time} (b), we use \textit{NeighborSize} = 8.
As shown in Fig.~\ref{fig:time}(a), larger \textit{NeighborSize} leads to longer sampling time. The average sampling time for \textit{NeighborSize} of 1, 2, 4, and 8 are 3, 4, 7, and 14 ms. 
Similarly, the increase of sampling instances, as shown in Fig.~\ref{fig:time}(b) also results in longer sampling time. The average sampling time for 2k, 4k, 8k, and 16k instances is 2, 5, 9, and 15 ms. 
It is important to note that graphs with higher average degrees, i.e., TW, RE, and OR, have longer sampling time, while the impact of graph sizes on sampling time is secondary.

\begin{figure}[t]
    \centering    
    \subfloat[2,000 instances.]{
        \hspace{-.05in}\includegraphics[width=0.9\linewidth]{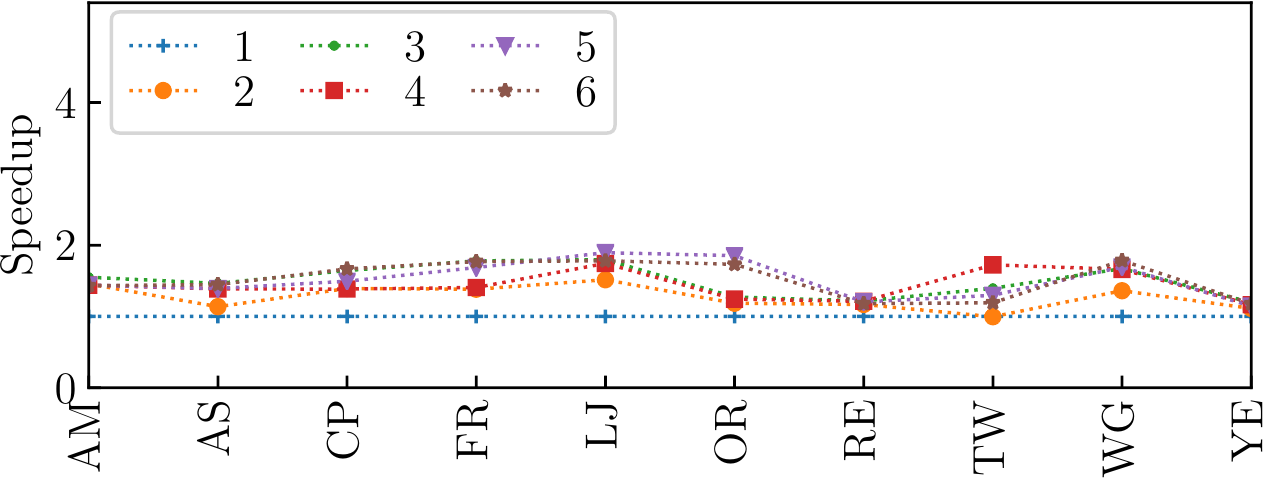}
    }
    \\
        \subfloat[8,000 instances.]{
        \hspace{-.05in}\includegraphics[width=0.9\linewidth]{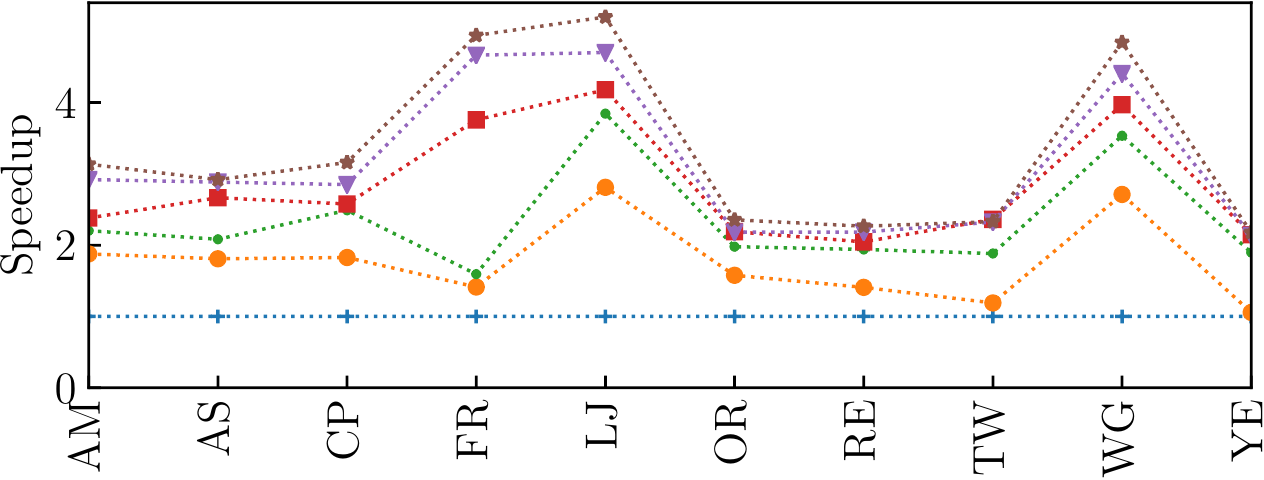}
    }
    
    \caption{\new{Scaling {\gsag} from 1 to 6 GPUs with (a) 2,000 and (b) 8,000 instances for biased neighbor sampling.\vspace{-0.15in}}
    }
    \label{fig:scaling}
\end{figure}

\vspace{0.02in}
\subsection{{\gsag} Scalability}
Fig.~\ref{fig:scaling} scales {\gsag} from 1 to 6 GPUs for different number of sampling instances.
For 2,000 and 8,000 instances, we achieve {1.8}$\times$ and {5.2}$\times$ speedup with 6 GPUs, respectively. The reason is that 2,000 instances fail to saturate 6 GPUs. With 8,000 instances, we observe more workloads that lead to better scalability. We also observe that lower average degree graphs present better scalability because their workloads are better distributed across sampling instances.

\section{Related Works}
\label{sect:related}

Despite there is a surge of frameworks for classical graph algorithms including think like a vertex~\cite{low2014graphlab,malewicz2010pregel}, an edge~\cite{roy2013x}, a graph~\cite{tian2013think}, an IO partition~\cite{liu2017graphene}, and Domain Specific Languages~\cite{hong2012green,zhang2018graphit}, among many others~\cite{liu2019simd,bader2006gtgraph,sundaram2015graphmat}, very few projects target graph sampling and random walk which are the focus of {\gsag}.
This section discusses the closely related work from the following three aspects.\vspace{3px} 
 
\textbf{Programming Interface.}
KnightKing~\cite{yang2019knightking} proposes a walker-centric model to support random walk~\cite{ribeiro2010estimating,li2015random}, e.g., Node2vec~\cite{grover2016node2vec,10.1145/3159652.3159706}, Deepwalk~\cite{perozzi2014deepwalk}, and PPR~\cite{ilprints750,lofgren2014fast,lin2019distributed}, and hence fails to accommodate sampling algorithms that are important for graph learning and sparsification~\cite{ribeiro2010estimating,gao2018large,chen2018fastgcn,ying2018graph,hamilton2017inductive,leskovec2006sampling,gaihre2019deanonymizing}. Similarly for~\cite{lakhotia2019parallel,chen2016general} which also only support limited sampling/random walk algorithms.  
GraphSAINT~\cite{zeng2019graphsaint,zeng2019accurate} explores three graph sampling methods, i.e., random vertex and edge sampling, and random walk based sampler, but fails to arrive at a universal framework.
\cite{tariq2017power} supports deletion based sampling algorithms~\cite{krishnamurthy2007sampling}.
But this design is inefficient for large graphs that need to remove most edges.
In this work, {\gsag} offers a bias-centric framework that can support both sampling and random walk algorithms, and hide the GPU programming complexity from end users.

\vspace{0.05in} 
\textbf{Transition Probability Optimizations.} Existing projects often explore the following optimizations, i.e., probability pre-computation and data structure optimization. Particularly, KnightKing~\cite{yang2019knightking} pre-computes the alias table for static transition probability, and resorts to dartboard for the dynamic counterpart which is similar to~\cite{zeng2019accurate}. Interestingly, kPAR~\cite{shi2019realtime} even proposes to pre-compute random walks to expedite the process.
Since large graphs cannot afford to index the probabilities of all vertices,~\cite{lin2019distributed} only pre-computes for hub vertices and further uses hierarchical alias method, i.e., alias tree for \cready{distributed} sampling.
However, not all sampling and random walk algorithms could have deterministic probabilities that support pre-computation.
{\gsag} finds \cready{inverse transform sampling} to be ideal for GPUs, and achieves superior performance over the state-of-the-art even when computing the probability during runtime.

\vspace{0.05in}
\textbf{Out-of-memory Processing.} GPU unified memory and partition-centric are viable method for out-of-memory graph processing. 
Since graph sampling is irregular, unified memory is not a suitable option~\cite{mishra2017um,li2019um}. Besides, partition-centric options~\cite{graphchi,zheng2015flashgraph,liu2017graphene,han2017graphie,chiang2019cluster} load each graph partition from either secondary storage to memory or CPU memory to GPU for processing. Since prior work deals with classical graph algorithms, they need BSP. 
In contrast, {\gsag} takes advantage of the asynchronous nature of sampling to introduce workload-aware scheduling and batched sampling to reduce the data transfer between GPU and CPU.

\section{Conclusion}
\label{sect:conclusion}

This paper introduces {\gsag}, a novel, generic, and optimized GPU graph sampling framework that supports a wide range of sampling and random walk algorithms. Particularly, we introduce novel bias-centric framework, bipartite region search and workload aware out-of-GPU and multi-GPU scheduling for {\gsag}.
Taken together, our evaluation shows that {\gsag} bests the state-of-the-art.

\section*{Acknowledgement}

We thank the anonymous reviewers for their helpful suggestions and feedbacks. This research is supported in part by the National Science Foundation CRII award No. 2000722, the U.S. Department of Energy, Office of Science, Office of Advanced Scientific Computing Research, under Contract No. DE-AC02-05CH11231 at Lawrence Berkeley National Laboratory, and Brookhaven National Laboratory, which is operated and managed for the U.S. Department of Energy Office of Science by Brookhaven Science Associates under contract No. DE-SC0012704.


	{
		\bibliographystyle{unsrt}
		\bibliography{references}
	}

\end{document}